\documentclass[11pt, oneside]{article}
\usepackage[margin=1in]{geometry}
\usepackage[utf8]{inputenc}

\title{Algorithms from Invariants: \\ Smoothed Analysis of Orbit Recovery over $SO(3)$}
\author{Allen Liu \thanks{Email: \texttt{cliu568@mit.edu}. This work was supported in part by an NSF Graduate Research Fellowship, a Fannie and John Hertz Foundation Fellowship and Ankur Moitra's NSF CAREER Award CCF-1453261 and NSF Large CCF1565235.}\and Ankur Moitra \thanks{Email: \texttt{moitra@mit.edu}. This work was supported in part by a Microsoft Trustworthy AI Grant, NSF CAREER Award CCF-1453261, NSF Large CCF1565235, a David and Lucile Packard Fellowship and an ONR Young Investigator Award.}}
\input{preamble.sty}
\begin{document}

\maketitle
\thispagestyle{empty}

\begin{abstract}
   In this work we study the orbit recovery problem over $SO(3)$, where the goal is to recover a function on the sphere from noisy, randomly rotated copies of it. Furthermore we assume that the function is a linear combination of low-degree spherical harmonics. This is a natural abstraction for the problem of recovering the three-dimensional structure of a molecule through cryo-electron tomography. When it comes to provably learning the parameters of a generative model, the method of moments is the standard workhorse of theoretical machine learning. It turns out that there is a natural incarnation of the method of moments for orbit recovery based on invariant theory. 
   
   Bandeira et al. \cite{bandeira2018estimation} used invariant theory to give tight upper and lower bounds on the sample complexity in terms of the noise level. However many of the key challenges remain: Can we prove bounds on the sample complexity that are polynomial in $n$, the dimension of the signal? The bounds in \cite{bandeira2018estimation} hide constants that have an unspecified dependence on $n$ and only hold in the limit as $\sigma^2 \rightarrow \infty$ where $\sigma^2$ is the variance of the noise. Moreover can we give efficient algorithms? 
   
  We revisit these challenges from the perspective of smoothed analysis, whereby we assume that the coefficients of the signal (in the basis of spherical harmonics) are perturbed by a small amount of Gaussian noise. Our main result is a quasi-polynomial time algorithm for orbit recovery over $SO(3)$ in this model. Our approach is based on frequency marching, which solves a linear system to find the higher degree coefficients assuming that the lower degree coefficients have already been found. Our main technical contribution is to show that these linear systems have unique solutions, are well-conditioned, and that the error can be made to compound over at most a logarithmic number of rounds. We believe that our work takes an important first step towards uncovering the algorithmic implications of invariant theory, particularly when fitting the parameters of a generative model with group symmetries.

\end{abstract}

\clearpage
\pagenumbering{arabic} 

\section{Introduction}

%\subsection{Background}

In this work we study the {\em orbit recovery} problem, where the goal is to recover a planted signal from noisy measurements under unknown group actions. Formally, there is
\begin{enumerate}
    \item[(1)] an unknown signal $x \in \mathbb{C}^n$
    
    \item[(2)] a group $G$ with a group action $\rho: G \rightarrow \mathbb{GL}(n, \mathbb{C})$
    
\end{enumerate}
\noindent and we get observations of the form $y = \rho(g) \cdot x + \eta$ where $g$ is drawn from the Haar measure on $G$ and $\eta$ is additive Gaussian noise with variance $\sigma^2$. The goal is to give a statistically and computationally efficient algorithm for recovering $x$ up to a group action. In particular, it is impossible to recover $x$ uniquely but we can still hope to find an element that is close to its orbit. We will be mainly interested in the case where $G = SO(3)$, which is challenging in part because $G$ is non-abelian.

Let us start with the motivation: Cryo-electron tomography (cryo-ET) \cite{frank2008electron} is a popular technique for imaging biological macromolecules and cells. It works by tilting an object and using an electron beam to image a two-dimensional slice through it. These two-dimensional slices can be combined to produce a three-dimensional image. However the image is extremely noisy. The goal is to develop methods for combining many noisy three-dimensional images to produce a single high-resolution image. The main difficulty is that the objects that are used to generate different samples are not necessarily aligned in some common starting configuration. Thus samples correspond to noisy observations of identical objects under unknown rotations. Orbit recovery over $SO(3)$ is a natural abstraction for this problem. In particular we assume that the object is described as an unknown signal on the sphere and moreover it is a linear combination of spherical harmonics of bounded degree. Let $x$ represent the coefficients in this basis. This is the natural notion of being band-limited, and ensures that the signal is finite-dimensional. Finally each sample is generated by randomly rotating the signal and getting noisy observations of the new coefficients in the basis of spherical harmonics. See Section~\ref{sec:setup} for further details. Later we will study generalizations where there are multiple band-limited functions, each associated with concentric spherical shells, which are all rotated together. 

Different flavors of orbit recovery arise naturally in other engineering problems. In signal processing, consider the problem of recovering a discrete and periodic signal from noisy and misaligned measurements. Orbit recovery over $\mathbb{Z}_n$ is a natural abstraction. It is also referred to as the discrete multireference alignment (discrete MRA) problem. Formally, there is an unknown signal $x$ which is a vector of dimension $n$. Each sample is generated by cyclically shifting the coordinates of $x$ by a random integer from $0$ to $n-1$ and adding Gaussian noise to each of the coordinates separately. The main challenge is that samples correspond to noisy observations under unknown cyclic shifts. Moreover if we can estimate $x$, we can also use $x$ to find an approximate alignment of the different samples. In other situations, we can allow the signal to be continuous and work with continuous cyclic shifts. This corresponds to orbit recovery over $SO(2)$. It is also called the continuous MRA problem. Formally, we assume the signal is band-limited. Let $x$ represent the coefficients of the signal in the Fourier basis. Now cyclically shifting the signal by $\theta$ corresponds to pointwise multiplying $x$ with a complex trigonometric polynomial. Finally we obtain noisy estimates of the new Fourier coefficients after the cyclic shift.

There are efficient algorithms for discrete and continuous MRA. 
%But observe that in both cases the underlying group is abelian. 
Perry et al. \cite{perry2019sample} gave an efficient algorithm for discrete MRA, and matching sample complexity lower bounds, based on tensor decompositions. Essentially, they exploit the fact that the samples can be viewed as coming from a mixture of $n$ spherical Gaussians in $n$ dimensions whose centers are cyclic shifts of each other. For continuous MRA there is an efficient algorithm based on frequency marching \cite{bendory2017bispectrum}. Moitra and Wein \cite{moitra2019spectral} gave an algorithm for heterogenous continuous MRA, i.e. where the samples are generated from a mixture model over $x$. However their algorithm only works when the $x$'s are random and solves the weaker list-recovery problem. 

In light of these works, a natural question is: Are there efficient algorithms for more general groups? The problem seems to become significantly harder\footnote{One explanation is that in the cases when $G = \mathbb{Z}_n$ or $G = SO(2)$, or whenever $G$ is finite and abelian, the invariant ring becomes {\em simple}: It can be generated by monomials. This will not be the case when $G = SO(3)$.} when the group is non-abelian and very little is known in this case.  An important work of Bandeira et al. \cite{bandeira2018estimation} established a link between orbit recovery and invariant theory.  However, even for the case of $SO(3)$, there are no known algorithms or even sample complexity guarantees that are sub-exponential in $n$.

\subsection{The Method of Moments and Invariant Theory}

%An important work of Bandeira et al. \cite{bandeira2018estimation} established a link between orbit recovery and invariant theory. It was known that the invariant ring determines $x$ up to its orbit under $G$ for compact $G$. They showed that the degree $d^*$ at which the invariant polynomials generate the full invariant ring determines the optimal sample complexity in the large noise regime. However there are some limitations to their results. First, actually determining $d^*$ is challenging. In the language of theoretical machine learning, in general for orbit recovery it is not known how many moments suffice for recovering the parameters. They were able to verify that $d^* = 3$ through computational algebra methods up to problems of size $15$ generically for orbit recovery over $SO(3)$. Second, their algorithms involve setting up and solving a large system of polynomial equations. The best known algorithms for this task require exponential time. However these systems have important algebraic structure to them. Is it possible to give computationally efficient algorithms with provable guarantees? Third, their statistical guarantees are asymptotic in nature because the bounds hide constants that depend ineffectively on $G$, $n$ and other quantities. In order to get sample complexity guarantees that are (quasi-)polynomial in all of the parameters of the problem we would need to give effective bounds on the stability of this system of polynomial equations. 

In this section, we give some background on the method of moments. Our main interest is in drawing natural parallels between the key steps in the method of moments, and the corresponding challenges in giving algorithms for orbit recovery via invariant theory. The method of moments is a workhorse in theoretical machine learning. It gives a general blueprint for fitting the parameters of a generative model, provided we can answer the following key questions:

\begin{enumerate}
    \item[(1)] \textbf{How many moments suffice to uniquely determine the parameters?} In many cases, the moments of a distribution can be expressed as polynomials in the unknown parameters. Then bounding the number of moments needed to uniquely identify the parameters is equivalent to reasoning about the solutions to some system of polynomial equations. 

    \item[(2)] \textbf{Is the system of equations, that determines the parameters in terms of the moments, stable to noise?} We are not given the exact moments of the distribution, but rather we estimate them given samples from the model. Thus we need to bound how the sampling error translates into error bounds on the parameters, so that we can get effective bounds on the sample complexity. 
    
    \item[(3)] \textbf{Are there efficient algorithms for solving the system of equations?} In general, solving systems of polynomial equations is computationally hard. However in some settings we can reduce a high-dimensional learning problem to a series of one-dimensional problems, so that the systems of polynomial equations each have a constant number of variables. In other cases, there are direct algorithms based on tensor decompositions, or even through rounding various semidefinite programming relaxations. 
    
\end{enumerate}

\noindent The method of moments has had many successes, including efficient algorithms for learning mixture models \cite{kalai2010efficiently, moitra2010settling, belkin2010polynomial, hsu2013learning, ge2015learning}, HMMs \cite{mossel2005learning}, phylogenetic trees \cite{mossel2005learning}, topic models \cite{arora2012learning, anandkumar2012spectral}, linear dynamical systems \cite{oymak2019non} and super-resolution \cite{moitra2015super, chen2021algorithmic}. It is also the building block of provably robust learning algorithms that can tolerate a constant fraction of their samples being arbitrarily corrupted \cite{diakonikolas2019robust, lai2016agnostic}. 

In orbit recovery, it is not immediately clear what sorts of polynomials in the unknown parameters can actually be estimated from the samples. This is because of the effect of the unknown group action. However, there is a natural incarnation of the method of moments for orbit recovery based on invariant theory. We say that a polynomial $q(x)$ for $x \in \mathbb{C}^n$ is invariant under the group action $G$ if for any $g \in G$ we have that 
\[ 
q(\rho(g) \cdot x) = q(x) \,.
\] 
This property allows us to construct unbiased estimators for the invariant polynomials (Lemma~\ref{fact:invariant-poly}) \cite{bandeira2018estimation}. Thus in orbit recovery, invariant polynomials are the natural moments of the generative model. 
And so we ask: {\em How much of the standard recipe for the method of moments can be carried over to orbit recovery?} 

The elegant paper of Bandeira, Blum-Smith, Kileel, Perry, Weed and Wein \cite{bandeira2018estimation} showed several striking applications of invariant theory to the orbit recovery problem. First let's give some additional background: The ring of invariant polynomials is called the {\em invariant ring}. A classic fact in invariant theory is that for any compact group $G$ that acts continuously, the invariant ring determines $x$ up to orbit \cite{kac1994invariant}. Thus for orbit recovery, the main question is: At what degree $d^*$ do the invariant polynomials of degree at most $d^*$ generate the full invariant ring? This is like bounding the number of moments in the method of moments. Essentially, both questions revolve around showing that there is no new information contained in higher degree moments/invariant polynomials.

The main result of Bandeira et al. \cite{bandeira2018estimation} is a tight bound on the sample complexity for the list recovery variant of orbit recovery over any compact group $G$. In particular they show that $\Theta(\sigma^{2d^*})$ samples are both necessary and sufficient. However the key questions from before remain unanswered. They were not able to determine $d^*$, except for small values of $n$ using computational algebra tools. Their statistical guarantees are asymptotic in nature because the hidden constants are ineffective and can depend on $G$ and $n$. In particular the upper bounds only hold in the limit as $\sigma^2 \rightarrow \infty$. And finally there are no known computationally efficient algorithms since the natural approach would be to solve a large system of polynomial equations with $n$ variables and one constraint for each of the generators. The connection between learning and invariant theory is tantalizing, but can we realize its promise by giving efficient algorithms whose running time and sample complexity are polynomial in $n$?

\subsection{Our Techniques}

Smoothed analysis is a popular framework for studying many learning and parameter recovery problems \cite{bhaskara2014smoothed, ge2015learning, huang2015minimal, bhaskara2020smoothed}. In particular, by assuming that the parameters are perturbed, we can circumvent many algebraic degeneracies that would otherwise derail the method of moments. In studying orbit recovery over $SO(3)$, we will also work in the framework of smoothed analysis: We assume that the coefficients of the signal in the basis of spherical harmonics have been perturbed by small Gaussian random variables. Our running time and sample complexity will have an inverse quasi-polynomial dependence on the size of the perturbation. Thus our algorithms work on all but a small measure of the parameter space around {\em any} starting point. For the related problem of orbit recovery over $SO(2)$, there are strong information-theoretic lower bounds that are known for the worst-case problem that come from the signal having non-trivial automorphisms \cite{bandeira2020optimal}. It is not known whether such algebraic degeneracies can arise for $SO(3)$. Nevertheless smoothing will be a fixture in our analysis.  We now give a high-level overview of our techniques.
%Add discussion on SO(2) smoothed analysis

Our algorithm will be based on frequency marching, which uses the degree three invariants and exploits the layered structure of the invariant polynomials. Specifically, if we assume that all lower degree spherical harmonics have already been found, we can write down constraints on the higher degree spherical harmonics. These constraints turn out to be systems of {\em linear} equations whose coefficients are polynomials in the lower degree spherical harmonics. Thus the invariant ring has a layered structure where we can add one layer at a time by solving a much simpler linear system. However the usual questions remain: Why do the linear systems have a unique solution? And how does the error, e.g. due to sampling noise, compound over the steps of the algorithm? 

The natural conjecture for orbit recovery over $SO(3)$ is that $d^* = 3$ \--- i.e. the degree three invariants determine the orbit of $x$ \cite{bandeira2020optimal}. We are not able to prove this conjecture, but we accomplish the next best thing: We prove that, after applying a perturbation, the degree three invariants uniquely determine how to extend any solution to a constant-sized subproblem to the full problem. In particular the linear systems that arise in frequency marching all have unique solutions, except in the first constant number of steps. This still suffices for our purposes since we can use standard techniques from algebraic geometry to analyze a brute-force algorithm to solve the constant-sized subproblem and then run frequency marching as usual (see Section~\ref{appendix:const-degree}).

The main technical challenge is to show that frequency marching is stable. Our approach is based on analyzing the condition number of the linear systems that arise in frequency marching. As usual, bounding the condition number of an algebraically structured linear system when its coefficients are perturbed is quite challenging and requires a tailor-made analysis \cite{bhaskara2014smoothed}. The linear systems that arise when solving for the $\ell$th order spherical harmonics have coefficients that are quadratic polynomials in the lower order spherical harmonics. These quadratic polynomials are given by the Clebsch-Gordan (CG) coefficients.  Roughly, the CG-coefficients are defined as the integral over the sphere of the point-wise product of three spherical harmonics.  Many properties of these coefficients are known. However, many other basic properties are, surprisingly, open.  For instance, what are the necessary and sufficient conditions for whether a CG coefficient is non-zero \cite{heim2009some}? Even if we could answer that, the non-zero coefficients are sometimes exponentially small. This makes it challenging to prove polynomial bounds on the condition number.

Nevertheless we are able to establish polynomial bounds on the condition number while only utilizing a few simple properties of CG coefficients (see Section~\ref{sec:CG-coeffs}): We exploit their orthogonality relations and estimate various edge-case CG coefficients that have simple explicit expressions (whereas most of the CG coefficients do not have nice expressions). Surprisingly, it turns out that this suffices to establish certain combinatorial properties about the locations of large CG coefficients (see e.g. Lemma~\ref{lemma:good-assignment}). Next we prove that the determinant of some square submatrix is bounded away from zero by combining the combinatorial properties of the locations of large CG coefficients with anti-concentration inequalities (see Lemma~\ref{lem:almostfullrank}, Section~\ref{sec:poly-anticoncentration}). Of course, our bound on the determinant must be exponentially small (because the determinant has high degree).  Nevertheless, we argue that this implies a lower bound on most of the singular values.  To obtain a bound on the condition number, we then argue that the remaining rows outside of this square submatrix suffice to give lower bounds on the remaining singular values with high probability.

%And finally we bootstrap this lower bound on the determinant to bound the condition number. 

%Another difficulty that we must contend with is that the matrix we are analyzing has entries that are quadratic polynomials (as opposed to linear) which makes it significantly more difficult to reason about even in a smoothed setting.  At a high level, our proof circumvents this issue as follows.  First, we prove that with the smoothing, the determinant of some square submatrix is bounded away from $0$.  Proving this involves reasoning about just a single polynomial and we can do this by combining the combinatorial properties of the locations of large CG coefficients with anti-concentration inequalities (see Lemma~\ref{lem:almostfullrank}, Section~\ref{sec:poly-anticoncentration}).   Of course, our bound on the determinant must be exponentially small (because the determinant has high degree).  Nevertheless, we argue that this implies a lower bound on the top $99 \%$ of the singular values.  To go from this to a bound on the condition number, we then argue that the remaining rows outside of this square submatrix suffice to fill in the remaining $1 \%$ of the singular values with high probability.   

Lastly we need to bound how the errors compound over different steps. In principle the error could grow exponentially with the number of steps of frequency marching. Instead we modify the standard frequency marching algorithm to take larger steps. In particular we solve for the coefficients at degree $\ell$ using only the coefficients of degree $k$ for $\ell/4 \leq k \leq 3 \ell/4$. We still get polynomial bounds on the condition number at each step. Finally we get quasipolynomial bounds on the sample complexity since the errors can only compound over a logarithmic number of levels.  %Note that our bounds on the condition number are actually polynomial, and the quasipolynomial complexity comes only from this compounding over logarithmically many rounds.

%Third, we bound the condition number of these linear systems. This is the most technically involved aspect of our analysis. The coefficients in the linear system come from the representation theory of $SO(3)$, specifically the Clebsch-Gordan coefficients. Many properties of these coefficients are known, such as necessary conditions for them to be non-zero. However necessary and sufficient conditions are not known \cite{heim2009some}. Even worse, many of the non-zero coefficients are exponentially small. Our proof involves finding well-behaved subsystems and bounding their condition number by exploiting the perturbations. Ultimately these bounds on the condition number are what allow us to get effective bounds on the sample complexity. There is still one remaining issue: The errors can compound at each step of the process, when we solve for the next layer of coefficients. We propose a modest modification of the standard frequency marching algorithm that takes larger steps. In particular we solve for the coefficients at layer $\ell$ using only the coefficients in layers $k$ for $\ell/4 \leq k \leq 3 \ell/4$. This allows us to get quasipolynomial bounds on the sample complexity.

\subsection{Our Results}

Putting it all together, our main result is:

%Our main result is a quasi-polynomial time algorithm to solve orbit recovery over $SO(3)$ \--- i.e. the cryo-electron tomography problem \--- for smoothed $x$. It turns out that there is an analogue of frequency marching for $SO(3)$ \cite{barnett2017rapid, bandeira2018estimation}. Roughly, the invariant polynomials at degree three have a triangular structure so that if we have solved for all the $(m', \ell')$ spherical harmonics with order less than $\ell$, we can set up a system of linear equations to solve for the ones at order $\ell$ and by counting arguments, there are enough of these equations. While this heuristic is one of the standard approaches in practice \cite{bendory2020single}, all that is known is that it succeeds on small sized problems (up to order $15$) \cite{bandeira2018estimation}, which can be checked by computer algebra techniques. 

%We define a variant of the frequency marching algorithm that takes longer strides. Roughly it solves for all coefficients of order at most $0.99 \ell$ and uses only these to solve for the coefficients at order $\ell$. This simple but important modification allows us to better control how the errors in the estimates propagate. Ultimately our analysis circumvents computer algebra techniques, which yield methods for checking whether a given system has a unique solution, and is instead based on identifying useful sub-systems within the ring of invariant polynomials and then analyzing their condition number in a smoothed analysis model. 

\begin{theorem} \label{thm:main-one-shell-informal} [informal]
Let $f: S^2 \rightarrow \C$ be a function that is a linear combination of spherical harmonics of degree at most $N$.  Also assume that the coefficients in this expansion are $\delta$-smoothed.  Let $\sigma$ be the noise level.  Then there is an algorithm that solves orbit recovery over $SO(3)$ with running time and sample complexity $(N/\delta)^{O(\log N)} (\sigma/\epsilon)^{O(1)}$ and recovers an $\epsilon$-approximation to $f$ up to rotation with $0.9$ probability.  
\end{theorem}

\noindent See Theorem \ref{thm:main-oneshell} for the full version. We also give extensions to multiple shells and the heterogeneous variant of the problem in Theorem \ref{thm:multiple-shells} and Theorem \ref{thm:heterogeneous} respectively. 

Another way to view our results is from the perspective of tensor decomposition, but with an underlying group action. 
In {\em (noisy) orbit tensor decomposition}, we get a noisy estimate of the tensor $$T = \int_{g \in G} (\rho(g) \cdot x)^{\otimes 3} dg$$ and our goal is to find $\widehat{x}$ so that $T$ and $\widehat{T}$ are close, where $$\widehat{T} = \int_{g \in G} (\rho(g) \cdot \widehat{x})^{\otimes 3} dg$$
In traditional tensor decompositions \cite{moitra2018algorithmic} there is no structure among the terms in the decomposition into rank one tensors. Even for orbit recovery over $\mathbb{Z}_n$, where we get tensors whose rank one terms are generated by vectors that are cyclic shifts of each other, we can ignore the algebraic structure among the factors and still get good enough bounds on the rank that we can apply tensor methods \cite{perry2019sample}. However for orbit tensor decomposition over a continuous group it is necessary to exploit this algebraic structure. One can view our main result as a stable algorithm that works in the specific case where $\rho(g)$ acts by applying an element of $SO(3)$ to a signal and writing down how the coefficients in the basis of spherical harmonics change. It turns out that we can interpret frequency marching as solving orbit tensor decomposition and we obtain:

%As before, note that the effect of applying $\rho(g)$ to $x$ has to do with how the coefficients in the basis of spherical harmonics change under a rotation of the band-limited signal. In contrast to traditional tensor decomposition problems, $T$ here is not necessarily low rank for infinite groups. Thus many of the standard tools, like Jennrich's algorithm \cite{moitra2018algorithmic} \--- that are the main workhorse behind provably learning various latent variable models \cite{anandkumar2014tensor} \--- fail. Instead, we need a new class of algorithms that can exploit the group structure. It turns out that we already have such an algorithm, because once again the answer is frequency marching and we can even re-use the same analysis as before. 

\begin{theorem}\label{coro:main-orbit-tensor-informal} [informal]
There is a quasi-polynomial time algorithm that solves noisy orbit tensor decomposition over $SO(3)$ with noise level $\sigma$ when $x$ is $\delta$-smoothed and its entries represent coefficients of spherical harmonics of degree at most $N$. In particular given an estimate of $T$ that is $(\delta/N)^{O(\log N)} \epsilon^{O(1)}$-close in Frobenius norm it outputs an $\widehat{x}$ so that $\widehat{T}$ is $\epsilon$-close to $T$ in Frobenius norm. Moreover the algorithm runs in $(N/\delta)^{O(\log N)} (1/\epsilon)^{O(1)}$ time and succeeds with $0.9$ probability.
\end{theorem}

\noindent See Theorem~\ref{coro:orbit-tensor} for the full version. We remark that the blow-up in approximation error that we incur when we solve the decomposition problem is consistent with our sample complexity bounds for orbit recovery, as when we take a quasi-polynomial number of samples we will be able to estimate the entries of $T$ to inverse quasi-polynomial accuracy. 

Improving the bounds to polynomial and also proving anything for the more general orbit tensor decomposition problem remain open. However we believe that our work takes an important next step, particularly when working with infinite and non-abelian groups. 

%Thus the layered structure of the invariant polynomials simultaneously solves both problems. Note that notions from invariant theory like the Hironaka decomposition give a blueprint for how to build up an invariant ring, but they seem incomparable to the structure we are exploiting here which is about how to algorithmically solve the associated systems of polynomial equations. 

\subsection{Relations to Cryo-Electron Microscopy}

Although our focus is on a different problem, we would be remiss to not mention cryo-electron microscopy (Cryo-EM). It is an imaging technique in structural biology that has been responsible for many important scientific discoveries. Its pioneers were awarded the 2017 Nobel Prize in Chemistry \cite{adrian1984cryo ,nogales2016development}. It involves taking two-dimensional images (tomographic projections) of a molecule with an unknown orientation and trying to reconstruct its three-dimensional structure. This reconstruction problem can be formulated as a generalized orbit retrieval problem  \cite{bandeira2018estimation} that involves not only a random group action, which is again a rotation, but also a projection. Giving statistically and computationally efficient algorithms for this problem is one of the major goals of the orbit retrieval literature. We hope that our work, which handles the case with rotations but no projections, might be a stepping stone towards this larger goal. 

There are other abstractions besides orbit recovery that are based on the idea that we can get noisy measurements of the relative rotation from one projection to another. This is called the synchronization approach \cite{singer2018mathematics}. However when the noise is large this approach is challenging to analyze, in part because consistent estimation of the group elements 
is impossible \cite{aguerrebere2016fundamental}.

\subsection{Paper Organization}

In Section~\ref{sec:setup}, we formally set up the problems that we study.  We define spherical harmonics and the corresponding group action of $SO(3)$.  We then formally state our main theorem (see Theorem~\ref{thm:main-oneshell}).  In Section~\ref{sec:invariants}, we formally introduce the invariant polynomial machinery that will be used in our algorithms. 
%Explicit expressions for the degree-$3$ invariant polynomials for $SO(3)$ are given in Section~\ref{sec:explicit-invariants}.  Recall that exploiting the layered structure of these invariant polynomials will be crucial for our algorithm.  

In Section~\ref{sec:algorithm-description}, we describe our algorithm in detail.  Recall that the two steps are to first solve a constant-sized subproblem and then solve the full problem by extending the solution to this constant-sized subproblem via frequency marching, exploiting the layered structure of the degree-$3$ invariant polynomials for $SO(3)$.  The analysis of the initial step is in Appendix~\ref{appendix:const-degree}.  The analysis of the frequency marching step is in Section~\ref{sec:analysis}.  We first note a few properties about the CG-coefficients (which appear as coefficients in the degree-$3$ invariant polynomials) in Section~\ref{sec:CG-coeffs}.  Recall that the properties that we need are orthogonality and explicit bounds for a few ``edge"-case coefficients.  Then in Section~\ref{sec:well-conditioned}, we prove that the linear systems that arise from frequency marching using the degree-$3$ invariants are well-conditioned (see Lemma~\ref{lem:well-conditioned}).  This is the main technical component.  We first prove certain combinatorial properties about the locations of large CG-coefficients in Lemma~\ref{lemma:good-assignment}.  We use these properties to lower bound the determinant of a square submatrix, and then combine with additional rows outside of this square submatrix to get a bound on the condition number of the linear system.

In Section~\ref{sec:multiple-shells}, we extend our results to the case of multiple shells, meaning that instead of just one function $f: S^2 \rightarrow \C$, there are actually several functions $(f^{(1)}, \dots , f^{(T)})$ all from $S^2 \rightarrow \C$ that are rotated together.  It turns out that our techniques extend naturally to this setting.  In Section~\ref{sec:heterogeneous}, we extend our results to the case of heterogeneous observations, meaning that each observation is one of $k$ possible functions $f^{[1]}, \dots , f^{[k]}$ with probabilities $w_1, \dots , w_k$.  In this setting, we show how to essentially decouple the mixture, using higher-order invariants to eventually estimate the degree-$3$ invariant polynomials of each individual function $f^{[i]}$ (see Section~\ref{sec:decouple}).  Finally, in Section~\ref{sec:further-discussion}, we discuss connections with other problems, namely tensor decomposition with group structure and cryo-electron microscopy.

\section{Preliminaries}\label{sec:setup}
In cryo-electron tomography (cryo-ET), there is an unknown function $f$ defined on the unit sphere in $\R^3$ i.e. $f: S^2 \rightarrow \C$.  We receive observations of the form
\[
y_i = R_i(f) + \zeta_i
\]
where $R_i \in \SO(3)$ is a random rotation, $R_i(f)$ is the function $x \rightarrow f(R_i^{-1}(x))$, and $\zeta_i$ is some noise function.  The goal is to recover the $f$ up to orbit under the action of $SO(3)$.

Of course, we need some assumption on $f$ that restricts it to a finite dimensional space in order to make the problem well-posed (as otherwise it is not even clear how to process a single observation).  The standard way to do this is through the concept of spherical harmonics, which form a basis of functions on the sphere and are the natural analogue of the Fourier basis but for a non-abelian group.

\subsection{Spherical Harmonics}

Here we introduce spherical harmonics and some of their key properties. See \cite{blanco1997evaluation} for more details. 
%We will not go into too much detail about spherical harmonics as we end up using surprisingly little information about their structure.  For a more detailed exposition, see \cite{blanco1997evaluation}.  We summarize the important properties below. 
We use the notation $(\theta, \phi)$ to denote the representation of a point in $S^2$ using spherical coordinates. 
\begin{definition}[Spherical Harmonics]
For integers $l \geq 0$ and $-l \leq m \leq l$, the spherical harmonic $Y_{lm}(\theta, \phi)$ is defined as 
\[
Y_{lm}(\theta, \phi) = (-1)^m N_{lm} P^m_l(\cos \theta) e^{im \phi}
\]
where 
\[
N_{lm} = \sqrt{\frac{(2l+1)(l-m)!}{4\pi (l + m)!}}
\]
and $P^m_l(x)$ are the associated Legendre polynomials defined as 
\[
P_l^m(x) = \frac{1}{2^l l!}(1 - x^2)^{m/2} \frac{d^{l+m}}{dx^{l + m}}(x^2 - 1)^l \,.
\]
\end{definition}

For a spherical harmonic $Y_{lm}$, its degree is $l$. Thus there are $2l + 1$ spherical harmonics of degree $l$.  Some key properties of the spherical harmonics are summarized below (see \cite{blanco1997evaluation}).

\begin{fact}\label{fact:harmonic-basic}
The spherical harmonics satisfy the following properties:
\begin{enumerate}
    \item The set of all spherical harmonics forms an orthonormal basis for square integrable functions on $S^2$
    
    \item For any rotation $R \in SO(3)$, $R(Y_{lm})$ can be written as a linear combination of $Y_{l(-l)}, \dots , Y_{ll}$.  
    
    \item Note the previous statement implies that the spherical harmonics of degree $l$ form a $2l+1$-dimensional representation of $SO(3)$.  In fact, these representations (over all $l$) are exactly the irreducible representations of $SO(3)$.
\end{enumerate}
\end{fact}

\subsection{Problem Formulation using Spherical Harmonics}\label{sec:problem-formulation}
We are now ready to formally define orbit recovery over $SO(3)$ using the language of spherical harmonics.  We assume that $f$ is band-limited i.e. it can be decomposed into spherical harmonics of degree at most $N$ for some parameter $N$.  Then learning $f$ is equivalent to learning its coefficients, say $\{f_{lm} \}_{l \leq N}$, in the decomposition of $f$ onto spherical harmonics i.e.
\[
f = \sum_{l \leq N}\sum_{m - -l}^l f_{lm} Y_{lm}(\theta, \phi) \,.
\]
The number of coefficients that we need to learn is 
\[
1 + 3 + \dots + (2N + 1) = (N+1)^2 \,.
\]
We will often view $f$ as a vector in $\C^{(N+1)^2}$ given by $\{f_{lm} \}_{l \leq N}$ since the representations are equivalent.  Since the spherical harmonics form an orthonormal basis, we have the following equality.
\begin{fact}\label{fact:orthonormality}[Parseval]
For all functions $f: S^2 \rightarrow \C$, 
\[
\norm{f}_2^2 = \sum_{l,m} |f_{lm}|^2 \,.
\]
\end{fact} 
We will assume that each of our observations consists of a function $\wh{f}$ represented as some set of coefficients $\{\wh{f}_{lm} \}_{l \leq N}$ obtained as 
\begin{equation}\label{eq:sample-model}
\wh{f}_{lm} = R(f)_{lm} + \zeta_{lm}
\end{equation}  
where $R \in SO(3)$ is a random unknown rotation and $\zeta$ is some noise.  Recall that by Fact \ref{fact:harmonic-basic}, a rotation $R$ acts as a separate linear transformation on $(f_{l(-l)}, \dots , f_{ll})$ for each $l$.  Thus, $R(f)_{lm}$ is a fixed linear combination (depending only on $R$) of $(f_{l(-l)}, \dots , f_{ll})$.  This also justifies the fact that we only need to consider the coefficients of $\wh{f}$ corresponding to spherical harmonics of degree at most $N$.  Furthermore it means that we can use the following notation: 
\begin{definition}
For $N' \leq N$, the function $f_{\leq N'}$ denotes truncating $f$ by keeping only the terms in its spherical harmonic expansion of degree at most $N'$.  Moreover it is important to note that for any $R \in SO(3)$,
\[
 R(f_{\leq N'})  = R(f)_{\leq N'} \,.
\]
\end{definition}

Of course, we can only hope to learn $f$ up to rotation.  Thus, our goal will be to output a function $\wt{f}$ with spherical harmonic coefficients $\{\wt{f_{lm}}\}_{l \leq N}$ of degree at most $N$ such that the error up to rotation is small.  More formally:
\begin{definition}
For two functions $f, \wt{f}$ whose spherical harmonic expansions have degree at most $N$, we define their distance up to rotation as 
\[
d_{SO(3)}(f, \wt{f}) = \min_{R \in SO(3)}\left(\norm{R(f) - \wt{f}}_2^2 \right) = \min_{R \in SO(3)}\left( \sum_{l = 0}^N \sum_{m = -l}^l |R(f)_{lm} - \wt{f_{lm}}|^2 \right) \,.
\]
Note the first equality makes sense even if the degree of the spherical harmonic expansion is larger than $N$.
\end{definition}

We also make a few remarks about the noise distribution for $\zeta_{lm}$.  The precise distribution of $\zeta_{lm}$ does not affect our algorithm or its analysis as long as the noise is bounded and unbiased.  Thus, we could equivalently consider the noise as being added directly to the function $f$ before measuring its spherical harmonics.  For simplicity though, we will assume that the real and imaginary part of  $\zeta_{lm}$ are drawn independently from $ N(0, \sigma^2)$ for each $l,m$ and the noise is added directly to the coefficient measurements.  We will assume $\sigma \geq 1$.

Finally, we will work in the framework of smoothed analysis.  Formally, an adversary picks the coefficients $\{ f_{lm}' \}$ and then the coefficients of the true function are smoothed by adding complex Gaussian noise
\[
\{f_{lm} \} =  \{ f_{lm}' + N(0, \delta^2) + i N(0, \delta^2) \} \,.
\]
where in the above, $i = \sqrt{-1}$.  After applying this procedure, we say that the coefficients are $\delta$-smoothed.

We are now ready to state our main theorem.
\begin{theorem}\label{thm:main-oneshell}
Let $f$ be a function $f: S^2 \rightarrow \C$ whose expansion in spherical harmonics $\{f_{lm} \}$  has degree at most $N$ and such that $\norm{f}_2 \leq  1$.  Assume that the coefficients $\{f_{lm} \}$ are $\delta$-smoothed.  Let $\eps > 0$ be the desired accuracy.  There is an algorithm (Algorithm \ref{alg:main}) that takes 
\[
Q = \left( \frac{N}{\delta}\right)^{O(\log N)} \poly\left(\sigma, \frac{1}{\eps} \right)
\]
noisy observations with noise drawn from $N(0, \sigma^2)$ entrywise, runs in $\poly(Q)$ time, and with probability $0.9$, outputs a function $\wt{f}$ whose expansion in spherical harmonics has degree at most $N$ and satisfies
\[
d_{SO(3)}(f , \wt{f}) \leq \eps \,.
\]
\end{theorem}

%There are 
%\[
%1 + 3 + \dots + 2N + 1 = (N+1)^2
%\]
%coefficients that we must learn.  Also note that we can only learn the coefficients in the decomposition of $f$ up to rotation.  More precisely, let $V_l = \R^{2l + 1}$.  Then our signal is an element of $V = \oplus_{l = 0}^{N} V_l$ and each element of $SO(3)$ 

\section{Invariant Polynomials}\label{sec:invariants}

Following Bandeira et al. \cite{bandeira2018estimation}, we will rely on an incarnation of the method of moments through invariant polynomials.  Recall that we are trying to recover the $(N+1)^2$-dimensional vector $\{f_{lm} \}_{l \leq N}$ from noisy observations.  We would like to measure quantities of the form $P(\{f_{lm} \}_{l \leq N})$ for various polynomials $P$ by averaging over the samples.  However for most polynomials $P$, this is not actually possible because all of our observations come with an unknown rotation.  However, there are certain polynomials that are rotationally invariant, which are called invariant polynomials.  We will be able to estimate these polynomials from samples and use them in our reconstruction algorithm.

We formally define invariant polynomials below.  We will define them over a general group $G$ that acts linearly on a vector space $V = \C^n$.  The exposition here will be brief and we omit many details.  A more detailed account can be found in \cite{kac1994invariant}.  For the purposes of this paper, we only need that the following results hold for $G = SO(3)$.  

\begin{remark}
Recall that in the introduction, we used the notation $\rho(g) \cdot x$ to emphasize the group action. 
%and the fact that for $SO(3)$, $\rho(g)$ is not just rotating $x$ but changing the coefficients in the basis of spherical harmonics.
However, to simplify notation later on, we will slightly abuse notation and just write $g \cdot x$.
\end{remark}

\begin{definition}
For a compact group $G$ acting linearly on a vector space $V = \C^n$, a polynomial $P \in \C[x_1, \dots , x_{n} ] $ is an invariant polynomial if for any element $g \in G$ and $x \in V$, 
\[
P( g \cdot x ) = P(x) \,.
\]
\end{definition}
\begin{remark}
This is equivalent to requiring that $P( g \cdot  x ) = P(x)$ as a formal identity where $x$ is a vector of $n$ formal variables and $g \cdot x$ denotes the linear transformation given by $g$ applied to $x$.
\end{remark}

The invariant polynomials form a ring, which is called the invariant ring.
\begin{definition}
The invariant ring $\C^{G}[x_1, \dots , x_{n} ]$ is the ring of all invariant polynomials.  For an integer $d$, let $\C_d^{G}[x_1, \dots , x_{n} ]$ denote the set of invariant polynomials that are homogeneous of degree $d$.  
\end{definition}

We now introduce the concept of the Reynold's operator which will be crucial in understanding the structure of the ring of invariant polynomials.  
\begin{definition}
Let $G$ be a compact group acting linearly on the vector space $V = \C^n$.  The Reynolds operator $\mcl{R}: \C[x_1, \dots , x_n] \rightarrow \C^G[x_1, \dots , x_n]$ is defined as 
\[
\mcl{R}(P)(x) = \E_{g \sim G}[P(g \cdot x)]
\]
where $x = (x_1, \dots , x_n)$ and the expectation is taken with respect to the Haar measure of $G$.  
\end{definition}

For the case of $G = SO(3)$, the Haar measure is simply the uniform measure.  It is immediate from the definition that $\mcl{R}(P)$ projects onto the invariant ring.  We now present several basic properties.  The proofs can be found in \cite{kac1994invariant}.

\begin{fact}\label{fact:invariant-poly}
The invariant polynomials satisfy the following properties:
\begin{enumerate}
    \item Any polynomial $P \in \C_d^{G}[x_1, \dots , x_{n} ]$ can be written as a linear combination 
    \[
    \sum_{\alpha} c_{\alpha} \mcl{R}\left(x^{\alpha}\right)
    \]
    where $\alpha$ ranges over all $n$-variate monomials of degree at most $d$ and $c_{\alpha} \in \C$ for all $\alpha$.
    
    \item The invariant ring $\C^G[x_1, \dots , x_n]$ is finitely generated (as an algebra over $\C$)
    
    \item For any $x, x'$ such that $g \cdot x \neq x'$ for all $g \in G$ (i.e. $x,x'$ are not in the same orbit), there is a polynomial $P \in \C^G[x_1, \dots , x_n]$ such that 
    \[
    P(x) \neq P(x') \,.
    \]
\end{enumerate}
\end{fact}

\subsection{Estimating Invariant Polynomials}
The reason invariant polynomials are useful is that we can measure their values even when our observations are noisy and come with an unknown rotation.  We now use a few simple results from \cite{bandeira2018estimation} to make this intuition quantitative.

\begin{lemma}[See Section 7.1 in \cite{bandeira2018estimation}]\label{lem:est-moments}
Let $G$ be a compact group acting linearly on a vector space $V = \C^n$.  Let $x \in V$ and assume $\norm{g \cdot x}_2 \leq 1$ for all $g \in G$.  Assume we are given $Q$ independent observations $y_1, \dots , y_Q$ of the form
\[
y_j = g_j \cdot x + N(0, \sigma^2 I) + iN(0, \sigma^2 I)
\]  
where $g_j$ is drawn randomly (according to the Haar measure) from $G$ and $i = \sqrt{-1}$.  

Let $P_{\alpha}(x) = x^{\alpha}$ for all $n$-variate monomials $x^{\alpha}$ of degree at most $d$.  Let $\tau > 0$ be a parameter.  We can compute in $\poly(Q, n^d) $ time, estimates $\wt{P_{\alpha}(x)}$ such that with probability $1- \tau$, we have for all $\alpha$,
\[
\left \lvert \wt{P_{\alpha}(x)} - \E_{g \sim G}\left[ P_{\alpha}(g \cdot x)  \right] \right \rvert \leq c_d \sigma^d \sqrt{\frac{\log n/\tau}{Q}}
\]
where $c_d$ is a constant depending only on $d$.
\end{lemma}
\begin{remark}
The results in \cite{bandeira2018estimation} are stated for functions over $\R$ instead of functions over $\C$.  To transfer them to $\C$, it suffices to separate complex variables into their real and imaginary parts.
\end{remark}

As an immediate consequence of the above, we can estimate all low-degree invariant polynomials using polynomially many samples.
\begin{lemma}\label{lem:est-invariant}
Let $G$ be a compact group acting linearly on a vector space $V = \C^n$.  Let $x \in V$ and assume $\norm{g \cdot x}_2 \leq 1$ for all $g \in G$.  Assume we are given $Q$ independent observations $y_1, \dots , y_Q$ of the form
\[
y_j = g_j \cdot  x + N(0, \sigma^2 I) + iN(0, \sigma^2 I)
\]  
where $g_j$ is drawn randomly (according to the Haar measure) from $G$ and $i = \sqrt{-1}$.  

Let $\eps$ be a desired accuracy parameter and $\tau$ be the allowable failure probability.  If
\[
Q \geq O_d(1) \poly\left( n^d, \sigma^d, \frac{1}{\eps}, \log \frac{1}{\tau}\right)
\]
the for any invariant polynomial $P$ of degree at most $d$ with coefficients of magnitude at most $1$, we can compute in $\poly(Q)$ time, an estimate $\wt{P(x)}$ such that with probability $1 - \tau$,
\[
\left \lvert \wt{P(x)} - P(x) \right \rvert \leq \eps  \,.
\]
\end{lemma}
\begin{proof}
Write $P$ as a sum of monomials i.e.
\[
P = \sum_{\alpha } c_{\alpha} x^{\alpha} \,.
\]
Since $P$ is an invariant polynomial, we must have
\[
\E_{g}[P(g \cdot x)] = P(x)
\]
so therefore
\[
P(x) = \sum_{\alpha} c_{\alpha} \E_{g} [P_{\alpha}(g \cdot x)]
\]
where $P_{\alpha} = x^{\alpha}$.  Now by Lemma \ref{lem:est-moments}, we can estimate each of the terms $\E_{g} [P_{\alpha}(g \cdot x)]$ to within $\eps/(n+1)^d$ so since there are at most $(n+1)^d$ different monomials and all of the coefficients have magnitude at most $1$, we are done by the triangle inequality.
\end{proof}

Now we briefly explain what happens in our setting where $G = SO(3)$.  Recall we have a function $f:S^2 \rightarrow \C$ whose expansion into spherical harmonics $\{f_{lm} \}_{l \leq N}$ has degree at most $N$.  An element $R \in SO(3)$ acts linearly on the $(N+1)^2$-dimensional vector given by $\{f_{lm} \}_{l \leq N}$.  In fact, by Fact \ref{fact:harmonic-basic}, it can be decomposed into $N + 1$ separate linear maps on $\{f_{00} \}, \dots , \{f_{N(-N)}, \dots , f_{NN} \} $ respectively.  Thus, Lemma \ref{lem:est-moments} and Lemma \ref{lem:est-invariant} can be applied to our sampling model (\ref{eq:sample-model}) by viewing the action of $R \in SO(3)$ as a linear operator on the $(N+1)^2$-dimensional vector $\{f_{lm} \}_{l \leq N}$.    

We will use Lemma \ref{lem:est-invariant} to measure polynomials $P(\{f_{lm} \}_{l \leq N})$ where $P: \C^{(N+1)^2} \rightarrow \C$ is a polynomial in $(N+1)^2$ variables.  We will slightly abuse notation and use $P(f)$ to denote the polynomial $P$ applied to the spherical harmonic coefficients of $f$.

\subsection{Explicit Formulas}\label{sec:explicit-invariants}
Here, we present explicit formulas from \cite{bandeira2018estimation} for the degree-$3$ invariant polynomials.  These formulas will be used later.  
\begin{theorem}[\cite{bandeira2018estimation}]\label{thm:explicit-invariants}
The degree-$3$ invariant polynomials are (up to scaling) 
\begin{equation}\label{eq:deg3-invariant}
\mathcal{I}_{l_1,l_2,l_3}(f) = \sum_{\substack{k_1 + k_2 + k_3 = 0 \\ \abs{k_i} \leq l_i}} (-1)^{k_3} \langle l_1k_1l_2k_2| l_3(-k_3) \rangle f_{l_1k_1}f_{l_2k_2}f_{l_3k_3}
\end{equation}
where $\langle l_1k_1l_2k_2| l_3(-k_3) \rangle$ is defined in Fact~\ref{fact:CG-explicit} and denote the Clebsch-Gordan (CG) coefficients.
\end{theorem}
\begin{fact}\label{fact:CG-trivial}
All CG-coefficients are real numbers and satisfy $\left \lvert \langle l_1k_1l_2k_2| l_3(-k_3) \rangle \right \rvert \leq 1$.  Also if $l_1, l_2 , l_3$ are not the sides of a triangle (i.e. $l_1 > l_2 + l_3$ or $l_1 < |l_2 - l_3|$) then $ \langle l_1k_1l_2k_2| l_3(-k_3) \rangle = 0$.
\end{fact}
\begin{remark}
In light of Fact \ref{fact:CG-trivial}, the invariant polynomials $\mcl{I}_{l_1,l_2,l_3}$ are nonzero only when $l_1, l_2,l_3$ are the sides of a triangle.
\end{remark}
We will define the Clebsch-Gordan-coefficients and prove several additional properties about them in Section \ref{sec:CG-coeffs}.  For now, note that the invariant polynomials are layered in the sense that if $l_1, l_2 < l_3$ and the spherical harmonic coefficients of degree $l_1, l_2$ are known, then the invariant $\mcl{I}_{l_1,l_2,l_3}$ is linear in the spherical harmonics of degree $l_3$.  Exploiting this structure is one of the key to our algorithm and its analysis.

\section{Algorithm}\label{sec:algorithm-description}

\subsection{Algorithm Overview}
We are now ready to present our main algorithm .  At a high-level, we might hope to measure the values of all low-degree invariant polynomials using Lemma \ref{lem:est-invariant} and then solve the resulting polynomial system for the values of the coefficients $f_{lm}$.  Naively, this would be computationally intractable since it would involve solving a polynomial system with many variables and equations.  Instead we exploit the layered structure of the degree-$3$ invariant polynomials. Let $C$ be a sufficiently large (universal) constant.  Our algorithm consists of two phases:
\begin{enumerate}
    \item Recovering all spherical harmonic coefficients of constant degree i.e. all $f_{lm}$ with $l \leq C$ where $C$ is a universal constant
    \item Using the layered structure to do frequency marching i.e. we iteratively solve for the spherical harmonic coefficients of degree $C+1, C+2 , \dots $ and so on
\end{enumerate}

\subsubsection{Recovering Constant-degree Coefficients}
We can grid search for the values of the spherical harmonic coefficients $f_{lm}$ with $l \leq C$ (since there are only $(C+1)^2$ variables and $C$ is a constant).  It remains to show that we can construct a test for our guesses such that the test passes only if the guess is close to the truth up to orbit.

Recall property $2$ of Fact \ref{fact:invariant-poly}.  Note that we can view the spherical harmonic coefficients of degree at most $C$ as a $(C+1)^2$-dimensional vector and rotations $\R \in SO(3)$ act linearly on this vector.  Thus, we can compute invariant polynomials $P_1, \dots , P_k$ that generate the invariant ring using standard techniques from computational algebra (see e.g. Section 8.1 in \cite{bandeira2018estimation}).  The actual algorithm for this step does not matter because $C$ is constant so all of this can be done in $O(1)$ time and all of $P_1, \dots ,P_k$ have degree $O(1)$.  Now we describe how to test our guesses.  Let our guesses be $f_{lm}'$ for $l \leq C$.  For a desired accuracy parameter $\eps$, to test whether our guess is close to the truth up to orbit, we evaluate 
\[
P_1\left(\{f_{lm}'\}_{l \leq C}\right), \dots , P_k\left(\{f_{lm}'\}_{l \leq C}\right)
\]
and check whether they are close to 
\[
P_1\left(\{f_{lm}\}_{l \leq C}\right), \dots , P_k\left(\{f_{lm}\}_{l \leq C}\right)
\]
which we can measure by Lemma \ref{lem:est-invariant}.  Using results from computational algebraic geometry (Theorem 7 in \cite{solerno1991effective}), we can prove that there is a universal constant $K$ such that if 
\[
\left \lvert P_j\left(\{f_{lm}'\}_{l \leq C}\right) - P_j\left(\{f_{lm}\}_{l \leq C}\right) \right \rvert \leq \frac{1}{\eps^K}
\]
for all $j = 1, 2, \dots , k$, then the guesses $\{f_{lm}'\}_{l \leq C}$ must be $\eps$-close to the true coefficients $\{f_{lm}\}_{l \leq C} $ up to the action of some rotation in $SO(3)$.  Thus, if we grid search using a sufficiently fine grid, we get an algorithm that runs in time $(1/\eps)^{O(1)}$ and learns the coefficients $\{f_{lm}\}_{l \leq C} $ to within $\eps$.

\subsubsection{Frequency Marching for Remaining Coefficients}
Now assume that we know all of the spherical harmonic coefficients $f_{lm}$ with $l \leq C$.  We can set up a linear system to solve for the degree-$(C+1)$ spherical harmonic coefficients as follows.  For all $l_1, l_2  \leq 0.9999C$,  we can measure the value of $\mcl{I}_{l_1,l_2, C+1}(f)$ (as defined in (\ref{eq:deg3-invariant})).  This gives us $\Omega(C^2)$ linear equations to solve for the $2C + 3$ variables $\{f_{(C+1),-(C+1)}, \dots  , f_{(C+1),(C+1)} \}$.  Thus, as long as we can prove that these linear equations are well-conditioned, then we will be able to solve for $\{f_{(C+1),-(C+1)}, \dots  , f_{(C+1),(C+1)} \}$.  Proving that these linear equations are well-conditioned is the main technical piece of this paper (see Section \ref{sec:well-conditioned}).  Once we have done this, we can repeat the same method to solve for the spherical harmonic coefficients of degree $C + 2$ and so on.

The reason that we ensure $l_1, l_2 \leq 0.9999C$ (instead of just $l_1, l_2 \leq C$) is that this limits how errors in our estimates propagate.  In particular, with this modification, we only need to track errors through $O(\log N)$ levels of recursion instead of $N$ levels.  This propagation of errors through $O(\log N)$ levels of recursion is the source of the quasi-polynomial bound in Theorem \ref{thm:main-oneshell}.

\subsection{Formal Algorithm Description}

We now formally describe our algorithm.  As mentioned above, it consists of two subroutines: recovering the constant degree coefficients via grid-search and then frequency marching to solve for the remaining coefficients layer-by-layer. Below, we will use $C$ to denote a sufficiently large (universal) constant.

\begin{algorithm}[H]
\caption{{\sc $SO(3)$ Reconstruction Algorithm} }
\begin{algorithmic} 
\State \textbf{Input:} Parameters $N$, $\delta$, $\sigma$, $\eps$ 
\State \textbf{Input:} $Q$ samples of $f$ from the sampling model (\ref{eq:sample-model}) where 
\begin{itemize}
    \item $f$ is an unknown function with spherical harmonic expansion of degree at most $N$
    \item $\norm{f}_2 \leq 1$
    \item The coefficients of $f$ are $\delta$-smoothed
    \item The number of samples is
    \[
    Q = \left( \frac{N}{\delta}\right)^{O(\log N)} \poly\left(\sigma, \frac{1}{\eps} \right) \,.
    \]
\end{itemize}

\State Run {\sc Learn Constant-Degree Coefficients} with parameters $\sigma, \gamma$ where $\gamma  = \left( \delta/N\right)^{O(\log N)} \eps^{O(1)}$ to obtain estimates $\{\wt{f_{lm}} \}_{l \leq C} $
\For{L = C+1 , \dots , N}
\State Run {\sc Frequency Marching with Long Stride} with parameters $\delta, \sigma, \gamma$ where $\gamma = \left( \delta/N\right)^{O(\log N)} \eps^{O(1)}$, index $L$ and estimates $\{\wt{f_{lm}} \}_{l \leq L-1} $ to obtain solution $\{\wt{f_{L(-L)}}, \dots , \wt{f_{LL}} \}$ 
\EndFor
\State \textbf{Output:} $\{ \wt{f_{lm}} \}_{l \leq N} $
\end{algorithmic}
\label{alg:main}
\end{algorithm}

\begin{algorithm}[H]
\caption{{\sc Learn Constant-Degree Coefficients} }
\begin{algorithmic} 
\State \textbf{Input:} Parameters $\sigma$, $\gamma$ 
\State \textbf{Input:} $Q = \poly(\sigma, 1/\gamma)$ samples of $f$ from the sampling model (\ref{eq:sample-model})
\State Compute invariant polynomials $P_1, \dots , P_k$ that generate 
\[
\C^{SO(3)}[ x_{00}, \dots , x_{C(-C)}, \dots , x_{CC}] \,,
\]
the invariant ring for coefficients of spherical harmonics of degree at most $C$

\State Obtain estimates $\wt{P_1}, \dots , \wt{P_k}$ for $P_1(\{f_{lm} \}_{l \leq C} ), \dots , P_k(\{f_{lm} \}_{l \leq C} )$ using Lemma \ref{lem:est-invariant}
\State Let $K$ be a sufficiently large universal constant (in terms of $C, P_1, \dots , P_k$)
\State Grid search for values of $\{f_{lm}\}_{l \leq C}$ with $|f_{lm}| \leq 1$ with discretization $(1/\gamma)^{O(K)}$
\For{each guess $\{\wt{f_{lm}}\}_{l \leq C}$}
\State Check if for all $j \in [k]$
\[
\left \lvert P_j( \{\wt{f_{lm}}\}_{l \leq C}) - \wt{P_j} \right \rvert \leq 0.2 (1/\gamma)^{K}
\]
\If{above check passes}
\State \textbf{Output:} $\{\wt{f_{lm}}\}_{l \leq C}$
\State \textbf{break}
\EndIf
\EndFor
\end{algorithmic}
\label{alg:const}
\end{algorithm}

\begin{algorithm}[H]
\caption{{\sc Frequency Marching with Long Stride} }
\begin{algorithmic} 
\State \textbf{Input:} Parameters $\delta, \sigma, \gamma$
\State \textbf{Input:} Index $L$ and estimates $\{ \wt{f_{lm}} \}_{l \leq L - 1}$ of coefficients of degree less than $L$ 
\State \textbf{Input:} $Q = \poly(L\sigma/(\delta \gamma))$ samples of $f$ from the sampling model (\ref{eq:sample-model})

\State 
\For{all integers $a,b \leq 0.9999 L$}
\State Compute estimate $\wt{\mcl{I}_{a,b,L}}$ of $\mcl{I}_{a,b,L}(f)$ using Lemma \ref{lem:est-invariant} 
\EndFor
\For{all integers $a,b \leq 0.9999 L$}
\State For variables $X = \{ x_{L(-L)}, \dots , x_{LL} \}$ , define the vector $M_{(a,b)}$ such that
\[
  M_{(a,b)} \cdot  X = \sum_{\substack{k_1 + k_2 + k_3 = 0 \\ \abs{k_1} \leq a, \abs{k_2} \leq b, \abs{k_3} \leq L}} (-1)^{k_3} \langle ak_1bk_2| L(-k_3) \rangle \wt{f_{ak_1}} \wt{f_{bk_2}} x_{Lk_3}
\]
\EndFor
\State Let $M$ be the $(0.9999L)^2 \times (2L + 1)$ matrix with rows given by $M_{(a,b)}$ for $a,b \leq 0.9999 L$
\State Let $\mcl{I}$ be the vector of length $(0.9999L)^2$ with entries $\wt{\mcl{I}_{a,b,L}}$ for $a,b \leq 0.9999 L$
\State \textbf{Output:} $\{\wt{f_{L(-L)}}, \dots , \wt{f_{LL}} \}$  as the solution to
\[
 \arg \min_{X} \left( \norm{MX - \mcl{I}}_2^2\right) \,.
\]
\end{algorithmic}
\label{alg:iterate}
\end{algorithm}

\subsection{Analysis of Algorithm \ref{alg:main}}

Naturally, the proof of Theorem \ref{thm:main-oneshell}, which involves analyzing Algorithm \ref{alg:main}, is split into two parts.  The first part involves analyzing {\sc Learn Constant-Degree Coefficients} and the second part involves analyzing {\sc Frequency Marching with Long Stride}.  The two main lemmas that we prove (one for each part) are stated below.

It will be important to note that the first part does not depend on the smoothing of the coefficients in the spherical harmonic expansion of $f$.
\begin{lemma}\label{lem:alg1-analysis}
With probability $1 - 2^{-1/\gamma}$ over the samples, the output of {\sc Learn Constant-Degree Coefficients} satisfies the property that there is a rotation $R \in SO(3)$ such that 
\[
\sum_{l = 0}^C \sum_{m = -l}^l |\wt{f_{lm}} - R(f)_{lm}|^2 \leq \gamma \,.
\]
This statement holds with no smoothing on the coefficients of $f$.
\end{lemma}

The second part does depend crucially on the smoothing of the coefficients in the spherical harmonic expansion of $f$.
\begin{lemma}\label{lem:alg2-analysis}
With probability $1 - 2^{-L^{0.1}}$ over the random smoothing of the coefficients, the following holds.  Given any initial estimates for  {\sc Frequency Marching with Long Stride} that satisfy 
\[
\sum_{l = 0}^{0.9999L} \sum_{m = -l}^l |\wt{f_{lm}} - f_{lm}|^2 \leq \gamma' 
\]
for some sufficiently small $\gamma' \leq (\delta / L)^{O(1)}$, the algorithm, with probability $1 - 2^{-L}$ over the samples, outputs a solution $\{\wt{f_{L(-L)}}, \dots , \wt{f_{(LL)}} \}$ satisfying
\[
\sum_{m = -L}^L |\wt{f_{Lm}} - f_{Lm}|^2 \leq \poly(L/\delta) (\gamma' + \gamma) \,.
\]
\end{lemma}

The proof of Lemma~\ref{lem:alg1-analysis} is given in Appendix~\ref{appendix:const-degree}.  The proof of Lemma~\ref{lem:alg2-analysis} is given in Section~\ref{sec:analysis}.  Now we can prove Theorem \ref{thm:main-oneshell} using Lemma \ref{lem:alg1-analysis} and Lemma \ref{lem:alg2-analysis}.

\begin{proof}[Proof of Theorem \ref{thm:main-oneshell}]
By Lemma \ref{lem:alg1-analysis}, with probability $0.99$, we recover coefficients $\{ \wt{f_{lm}} \}_{l \leq C}$ such that there is a rotation $R \in SO(3)$ such that
\[
\sum_{l = 0}^C \sum_{m = -l}^l |\wt{f_{lm}} - R(f)_{lm}|^2 \leq \gamma \,,
\]
where $\gamma = \left( \delta/N\right)^{O(\log N)} \eps^{O(1)}$.

Now without loss of generality we may pretend that the hidden function is actually $R(f)$ and repeatedly apply Lemma \ref{lem:alg2-analysis}.  Applying Lemma \ref{lem:alg2-analysis} for $L = C, C+1 , \dots 1.0001C$, we get that
\[
\sum_{l= 0}^{1.0001C}\sum_{m = -l}^l  |\wt{f_{lm}} - R(f)_{lm}|^2  \leq \poly(L/\delta) \gamma \,.
\]
Now we can repeat the above argument for $L = 1.0001C, \dots , (1.0001)^2C$.  Overall, we repeat this argument at most $O(\log N)$ times before we get to $L = N$.  We conclude that with probability at least
\[
1 - (2^{-C^{0.1}} + 2^{-C}) + \dots + (2^{-N^{0.1}} + 2^{-N}) \geq 0.99
\]
(since $C$ is a sufficiently large universal constant), we have that
\[
\sum_{l= 0}^{N}\sum_{m = -l}^l |\wt{f_{lm}} - R(f)_{lm}|^2  \leq \poly(N/\delta)^{O(\log N)} \gamma \leq \eps 
\]
and we are done.  Note it is clear that the total number of samples used and the total runtime of both subroutines is bounded by $(N/\delta)^{O(\log N)} \poly\left(\sigma, 1/\eps \right) $.
\end{proof}

\section{Analysis of {\sc Frequency Marching with Long Stride}}\label{sec:analysis}
In this section, we prove Lemma \ref{lem:alg2-analysis}.  The first step will be to better understand the CG-coefficients that appear in the system at the crux of the algorithm. 

\subsection{Properties of CG-coefficients}\label{sec:CG-coeffs}

First, we note that the CG-coefficients admit a not-so-simple explicit formula (see \cite{bandeira2018estimation}).

\begin{fact}[Explicit Formulas]\label{fact:CG-explicit}
\begin{align*}
\langle l_1 m_1 l_2 m_2 | l m \rangle = &1_{m_1 + m_2 = m} \sqrt{ \frac{(2l+1)(l + l_1 - l_2)!(l - l_1 + l_2)!(l_1 + l_2 - l)!}{(l_1 + l_2 + l + 1)!}}     \\
 &\times \sqrt{(l+m)! ( l - m)! (l_1 - m_1)! (l_1 + m_1)! (l_2 - m_2)! (l_2 + m_2)!}  \\
 & \times\sum_{k}\frac{(-1)^k}{k!(l_1 + l_2 -l + k)!(l_1 - m_1 - k)!(l_2 + m_2 - k)!(l - l_2 + m_1 + k)!(l - l_1 - m_2 + k)!}
\end{align*}
where the sum in the above expression is over all integers $k$ for which the factorials in the summand are defined.
\end{fact}

The next important fact to note is that the CG-coefficients satisfy several symmetry and orthogonality relations (again see \cite{bandeira2018estimation}).

\begin{fact}[Orthogonality Relations]\label{fact:CG-orthogonality}
The CG-coefficients satisfy the following relations
\begin{enumerate}
    \item Symmetry
    \begin{align*}
    \langle l_1k_1 l_2k_2| lk \rangle     &= (-1)^{l_1 + l_2 - l}\langle l_2k_2 l_1k_1| lk \rangle \\
    &= (-1)^{l_1 - k_1} \sqrt{\frac{2l + 1}{2l_2 + 1}} \langle l_1k_1l(-k)| l_2(-k_2) \rangle
    \end{align*}
    \item Orthogonality type-1
    \begin{align*}
        \sum_{l,k} \langle l_1k_1 l_2k_2| lk \rangle\langle l_1k_1' l_2k_2'| lk \rangle = 1_{k_1 = k_1'}1_{k_2 = k_2'}
    \end{align*}
    \item Orthogonality type-2
    \begin{align*}
        \sum_{k_1,k_2} \langle l_1k_1 l_2k_2| lk \rangle\langle l_1k_1 l_2k_2| l'k' \rangle = 1_{l = l'}1_{k = k'}
    \end{align*}
\end{enumerate}
\end{fact}

\subsubsection{CG-coefficients when $l_1 = m_1$}

The CG-coefficients are significantly easier to understand when $l_1 = m_1$ because the sum only contains one term (for $k = 0$).  Understanding these coefficients will be a crucial part of our algorithm.  Throughout this section, we will assume that $l$ is sufficiently large and restrict ourselves to considering $l_1,l_2,l$ satisfying the following conditions: 
\begin{itemize}
    \item $0.01l < l_1 < 0.0101 l$
    \item $0.9901 l < l_2 < 0.9999l$
\end{itemize}

The first claim tells us that for certain $k_2,k$, the CG-coefficient  $\langle l_1 l_1  l_2 k_2 | l k \rangle$ is at least inverse-polynomially large.  
\begin{claim}\label{claim:large-CG-coeffs}
Let $l_1, l_2, l$ satisfying $0.01l < l_1 < 0.0101 l$ and $0.9901 l < l_2 < 0.9999l$ be given.  Let  
\[
k = \left\lfloor \frac{l^2 + l_1^2 - l_2^2}{2l_1} \right\rfloor \,.
\]
Then
\[
\abs{\langle l_1 l_1  l_2 (k - l_1)| l k \rangle} \geq \frac{1}{\poly(l)} \,.
\]
\end{claim}
\begin{proof}

Note
\begin{equation}\label{eq:sum_bound}
\sum_{k}\langle l_1 l_1  l_2 (k - l_1)| l k \rangle^2 = \sum_k \frac{2l+1}{2l_1 + 1} \langle l (-k)  l_2 (k - l_1)| l_1 (-l_1) \rangle^2 = \frac{2l+1}{2l_1 + 1}
\end{equation}
where we first used the symmetry properties and then orthogonality.

\noindent Also note that by Fact \ref{fact:CG-explicit}, we can write
\begin{align*}
\langle l_1 l_1  l_2 (k - l_1)| l k \rangle = &\sqrt{(2l+1)\frac{(l + l_1 - l_2)! (l - l_1 + l_2)! (l_1 + l_2 - l)!}{(l_1 + l_2 + l + 1)!}}  \\
&\times\sqrt{(l +k)! ( l-k)!(2l_1)!(l_2 - k + l_1)!(l_2 - l_1 + k)!}  \\
&\times \frac{1}{(l_1 + l_2 - l)!(l_2 + k - l_1)!(l - l_2 + l_1)!(l -k )! } \,.
\end{align*}

The above rearranges into 
\[
\langle l_1 l_1  l_2 (k - l_1)| l k \rangle = \sqrt{\frac{(2l+1)(l - l_1 + l_2)!(l+k)!(2l_1)!(l_2 - k + l_1)!}{(l_1 + l_2 + l + 1)!(l_1 + l_2 - l)!(l_2 + k - l_1)!(l - l_2 + l_1)! (l-k)! }} \,.
\]
Note that the above quantity is defined for integers $l_1, l_2, k$.  We will now extend it to real-valued inputs and understand the resulting real-valued function.  We then argue about what this means for the function restricted to the integers.  Let $f(x) = x \log x$.  Using Stirling's approximation,
\begin{align*}
\log \langle l_1 l_1  l_2 (k - l_1)| l k \rangle = &O(\log l) + f(l - l_1 + l_2) + f(l + k) + f((2l_1) + f(l_2 - k + l_1) \\ & - f(l_1 + l_2 + l) - f(l_1 + l_2 - l)- f(l_2 + k - l_1) - f(l - l_2 + l_1) - f(l-k) \,.
\end{align*}
Thus, it suffices to understand the function
\[
g(k) = f(l + k) + f(l_2 - k + l_1) -   f(l-k) - f(l_2 + k - l_1)\,.
\]
Note that the second derivative is
\begin{equation}\label{eq:second-deriv}
g''(k) = \frac{1}{(l + k)} + \frac{1}{ l_2 - k + l_1} - \frac{1}{l - k} - \frac{1}{l_2 + k - l_1} \,.
\end{equation}
Since $l_1 + l_2 > l$, it is immediately verified that this function is concave down.  Thus, $g(k)$ is maximized when its derivative is equal to $0$ i.e. when
\[
(l+k)(l-k) = (l_2 - k + l_1)(l_2 + k - l_1)
\]
which rearranges into 
\[
k = \frac{l_1^2 + l^2 - l_2^2}{2l_1} \,.
\]
It can be immediately verified that $c l < k < (1-c)l$ for some fixed constant $c > 0$.  Also note that taking the floor of $k$ affects the value of  
\[
f(l + k) + f(l_2 - k + l_1) -   f(l-k) - f(l_2 + k - l_1)\,.
\]
by at most a $O(\log l)$ factor.  Equation (\ref{eq:sum_bound}) implies that some CG-coefficients must be at least $1/\poly(l)$ and combining with the previous observation completes the proof.
\end{proof}

Actually, the proof of Claim \ref{claim:large-CG-coeffs} can be extended to give a slightly stronger statement, that for $k$ with 
\[
\left\lvert k - \frac{l_1^2 + l^2 - l_2^2}{2l_1} \right\rvert = O(\sqrt{l}) \,,
\]
the CG-coefficient $\langle l_1 l_1  l_2 (k - l_1)| l k \rangle$ is also inverse-polynomially large.
\begin{corollary}\label{corollary:large-CG-coeffs}
Let $l_1, l_2, l$ satisfying $0.01l < l_1 < 0.0101 l$ and $0.9901 l < l_2 < 0.9999l$ be given.  Let  $k$ be such that 
\[
\left\lvert k - \frac{l_1^2 + l^2 - l_2^2}{2l_1} \right\rvert  \leq \sqrt{l} \,,
\]
Then
\[
\abs{\langle l_1 l_1  l_2 (k - l_1)| l k \rangle} \geq \frac{1}{\poly(l)} \,.
\]
\end{corollary}
\begin{proof}
The proof is essentially the same as the proof of Claim \ref{claim:large-CG-coeffs}.  The only additional step needed is as follows.  Note that the derivative of the function
\[
g(k) = f(l + k) + f(l_2 - k + l_1) -   f(l-k) - f(l_2 + k - l_1)\,.
\]
(where $f(x) = x \log x$) is $0$ at 
\[
k_0 = \frac{l_1^2 + l^2 - l_2^2}{2l_1} \,.
\]
Also the second derivative, given by (\ref{eq:second-deriv}), is equal to $O(1/l)$ for the range of parameters that we are allowed.  Thus,
\[
g(k) - g(k_0) \leq O(1/l) \cdot |k - k_0|^2 = O(1)
\]
and combining with the result of Claim \ref{claim:large-CG-coeffs} we get the desired bound.
\end{proof}

\subsection{Bounding The Smallest Singular Value}\label{sec:well-conditioned}

Recall that in {\sc Frequency Marching with Long Stride} we want to solve for the spherical harmonic coefficients of degree $L$ and we take as inputs estimates for the lower-degree coefficients  $\{\wt{f_{lm}} \}_{l \leq L - 1}$.  The key step is solving a linear system in variables $X = \{x_{L(-L)}, \dots , x_{LL} \}$ of the form
\begin{equation}\label{eq:minimization}
\min_X \norm{MX - \mcl{I}}_2^2 
\end{equation}
where $M$ is a $(0.9999L)^2 \times (2L + 1)$ matrix where
\begin{itemize}
    \item The rows of $M$ are indexed by pairs of integers $a,b$ with $a,b \leq 0.9999L$ and the columns are indexed by integers $m = -L, \dots , L$ with entries given by 
    \[
    M_{(a,b) m} = \sum_{\substack{k_1 + k_2 = -m \\ \abs{k_1} \leq a, \abs{k_2} \leq b}} (-1)^{m} \langle ak_1bk_2| L(-m) \rangle \wt{f_{ak_1}} \wt{f_{bk_2}} 
    \]
   where $\wt{f_{lm}}$ (with $l \leq L - 1$) are the estimates for the lower-degree spherical harmonic coefficients. 
    \item The entries of $\mcl{I}$ are indexed by pairs of integers $a,b$ with $a,b \leq 0.9999L$ and are equal to our estimates $\wt{\mcl{I}_{a,b,L}}$ for the invariant polynomial $\mcl{I}_{a,b,L}(f)$
\end{itemize}

\begin{definition}\label{def:key-matrix}
Let $M^{\textsf{truth}}$ be the matrix whose rows are indexed by pairs of integers $a,b$ with $a,b \leq 0.9999L$ and the columns are indexed by integers $m = -L, \dots , L$ with entries given by 
\[
M^{\textsf{truth}}_{(a,b) m} = \sum_{\substack{k_1 + k_2 = -m \\ \abs{k_1} \leq a, \abs{k_2} \leq b}} (-1)^{m} \langle ak_1bk_2| L(-m) \rangle f_{ak_1} f_{bk_2}
\]
\end{definition}
\begin{remark}
Note $M$ would be equal to $M_{\textsf{truth}}$ if our input estimates were exactly correct.
\end{remark}

The key ingredient in the proof of Lemma \ref{lem:alg2-analysis} is to prove that $M^{\textsf{truth}}$ is reasonably well-conditioned (note this is equivalent to lower-bounding its smallest singular value because the largest singular value is trivially bounded above).  Once we prove this, the proof of Lemma \ref{lem:alg2-analysis} will be straight-forward. This is because if our estimates for the lower-degree coefficients  $\{\wt{f_{lm}} \}_{l \leq L - 1}$ and the invariant polynomials $\wt{\mcl{I}_{a,b,L}}$ are close to the truth, then $M$ will be well-conditioned as well.  Then, note that setting $x_{Lm} = f_{Lm}$ for all $m = -L, \dots , L$ would make the quantity $\norm{MX - \mcl{I}}_2^2 $ small. Thus the actual solution to (\ref{eq:minimization}) will be close to $\{f_{L(-L)}, \dots , f_{LL} \}$. 

The key lemma is stated below: 
\begin{lemma}\label{lem:well-conditioned}
There is an absolute constant $K$ such that for any $0< \eps < (\delta/L)^K$, with probability at least 
\[
1 - \eps^{\Omega(L^{0.4})} \,,
\] 
(over the $\delta$-smoothing of the coefficients), the smallest singular value of $M^{\textsf{truth}}$ is at least $\eps$.
\end{lemma}

The difficulty in proving Lemma \ref{lem:well-conditioned} lies in the fact that in the matrix $M^{\textsf{truth}}$, each spherical harmonic coefficient $ f_{ak_1}$ actually appears in many rows and thus we cannot easily decouple the randomness.  This makes it difficult to employ standard approaches such as leave-one-out distance (see e.g. \cite{bhaskara2014smoothed}) that require decoupling the randomness completely.

To get around this issue, we adopt a different approach.  Our approach can be separated into two parts.  We will first prove that a submatrix of $M^{\textsf{truth}}$ obtained by taking only a subset of the rows is robustly almost full-rank i.e. it has $0.95(2L + 1)$ non-negligible singular values.  We do this via algebraic techniques, arguing about the determinant as a polynomial and then using anticoncentration results in Appendix \ref{sec:poly-anticoncentration}.  While our bound on the determinant will be exponentially small in $L$, this will be enough to imply an inverse polynomial bound for the top $0.95(2L + 1)$ singular values.  we will then argue that adding in the remaining rows will make $M$ robustly full-rank i.e. it will have $2L + 1$ non-negligible singular values.  This second step can be done using a more standard approach because we only need to find a small amount of additional, independent randomness to go from $0.95(2L + 1)$ to $2L + 1$.

More formally, define the sets $A,B$ as follows:
\begin{itemize}
    \item $A$ denotes the set of integers between $0.01L$ and $0.0101 L$.
    \item $B$ denotes the set of integers between $0.9901 L$ and $0.9999 L$.
\end{itemize}
Let $M^{(A,B)}$ denote the matrix $M^{\textsf{truth}}$ restricted to the rows indexed by $(a,b) \in A \times B$.  We will first prove in Section \ref{sec:almostfullrank} that $M^{(A,B)}$ has $0.95(2L + 1)$ non-negligible singular values.  Once we have proved that $M^{(A,B)}$ is robustly almost full-rank, we then argue that we can use rows indexed by $(a,b)$ for $a \notin A, b \notin B$ to ensure that the entire matrix $M^{\textsf{truth}}$ is robustly full-rank.  This step is actually not too difficult because we only need to add $0.05(2L+1)$ more linearly independent rows and we can actually do this by finding a set of rows where the randomness in the smoothing is completely independent.

\subsubsection{ A Submatrix is Robustly Almost Full-rank}\label{sec:almostfullrank}

We will first prove that the determinant of some $0.95(2L+1) \times 0.95(2L+1)$-submatrix of $M^{(A,B)}$, viewed as a polynomial in the variables $\{ f_{ak_1}\}, \{f_{bk_2} \}$ is nonzero.  To do this, we will plug in values for the variables $\{ f_{ak_1}\}, \{f_{bk_2} \}$ and evaluate the determinant.  We will then argue about what this means for the determinant, when viewed as a formal polynomial.

We will restrict ourselves to a certain class of submatrices and a very simple class of assignments for the values of the variables $\{ f_{ak_1}\}, \{f_{bk_2} \}$.  These two notions are defined below.

\begin{definition}
We say that a subset $S$ of the rows of $M^{(A,B)}$ is $C$-\textbf{balanced} for a constant $C$ if 
\begin{itemize}
    \item For each $a \in A$, at most $C$ of the rows in $S$ are indexed by $(a,b')$ for some $b' \in B$.
    \item For each $b \in B$, at most $C$ of the rows in $S$ are indexed by $(a', b)$ for some $a' \in A$.
\end{itemize}
We say a submatrix $M'$ of $M^{(A,B)}$ is $C$-balanced if the subset of rows in $M'$ is $C$-balanced.
\end{definition}

\begin{definition}
We say that an assignment of values to the variables $\{ f_{ak_1}\}, \{f_{bk_2} \}$ is \textbf{simple} if 
\begin{itemize}
    \item For each $a$, there is exactly one value of $k_1$ with $-a \leq k_1 \leq a$, $f_{ak_1} \neq 0$
    \item For each $b$, there is exactly one value of $k_2$ with $-b \leq k_2 \leq b$, $f_{bk_2} \neq  0$
    \item For each $a$, either $f_{aa} = 1$ or $f_{a(-a)} = 1$ 
    \item For each $b$, there is some $-b \leq k_2 \leq b$ such that $f_{bk_2} = 1$
\end{itemize}
\end{definition}

Now we prove the following combinatorial lemma that will then allow us to argue about the determinant of some square submatrix of $M^{(A,B)}$.
\begin{lemma}\label{lemma:good-assignment}
There exists a simple assignment of values to the variables  $\{ f_{ak_1}\}, \{f_{bk_2} \}$ such that we can find a $O(L^{0.6})$-balanced square submatrix $M'$ of $M^{(A,B)}$ of size at least $0.95(2L+1) \times 0.95(2L+1)$ such that:
\begin{itemize}
    \item Each row of $M'$ contains exactly one nonzero entry
    \item Each column of $M'$ contains exactly one nonzero entry
    \item The nonzero entries have magnitude at least $1/\poly(L)$
\end{itemize}
\end{lemma}
\begin{proof}
We will actually take a uniformly random simple assignment and prove that with positive probability we can find a submatrix $M'$ of $M^{(A,B)}$ with the desired properties.  Here uniformly random assignment means 
\begin{itemize}
    \item For each $a \in A$ we will pick exactly one of $f_{a(-a)}$ and $f_{aa}$ to set to $1$ and we set all of the other $f_{ak_1}$ to $0$.
    \item For each $b \in B$ we pick exactly one of $f_{b(-b)}, f_{b(-b+1)}, \dots, f_{bb}$ to set to $1$ and set all of the others to $0$.
\end{itemize}    
These choices are all made independently and uniformly at random (over $2$ choices for each $a \in A$ and over $2b+1$ choices for each $b \in B$).

Note that when plugging a simple assignment of values into the matrix $M^{(A,B)}$, each row contains at most one nonzero entry.  First we will establish the following property.

\begin{property}\label{prop:statement1} For any given $k$ with $ 0.02L < k < 0.98L$ or $ -0.02L < k < -0.98L$, with at least $1 - \exp(-\Omega(L^{0.1}))$ probability (over the random assignment of values to the $f_{ak_1},f_{bk_2}$), there are $\Omega(L^{0.4})$ disjoint pairs  $(a,b)$ with $a\in A, b \in B$ such that there exist $k_1,k_2$ with
\begin{itemize}
    \item $f_{ak_1}f_{bk_2} = 1$
    \item $\abs{\langle ak_1bk_2| Lk \rangle} \geq \frac{1}{\poly(L)}$
\end{itemize}
\end{property}

To see this, without loss of generality $ 0.02L < k < 0.98L$.  From now on, we treat $k$ as fixed.   With $1 - \exp(-\Omega(L))$ probability, there are at least $|A|/3 $ elements $a \in A$ with $f_{aa} = 1$.  We use $A^+$ to denote this subset of $A$.  Divide the set $A^+$ into disjoint subsets of size $\sim 1000 L^{0.6}$, say $A_1, \dots , A_d$ where $d = \Omega(L^{0.4})$.

For each $i$ with $1 \leq i \leq d$, we say a pair of integers $(a,b)$ is $A_i$-forbidden if $ a \in A_i, b \in B$ and
\[
\left\lvert k - \frac{a^2 + L^2 - b^2}{2a} \right\rvert \leq \sqrt{L} \,.
\]
If for each $i$, there is $f_{bk_2} = 1$ where $k_2 = k - a$ for some $A_i$-forbidden pair $(a,b)$, then by Corollary \ref{corollary:large-CG-coeffs},
\[
\abs{\langle aab(k-a)| Lk \rangle} \geq \frac{1}{\poly(L)}
\]
and we would get the two desired conditions for some $a \in A_i$.  Next observe that for fixed $b$, the two desired conditions cannot hold for two distinct $a,a'$ since $f_{bk_2} = 1$ for only one value of $k_2$ and then we must have $a+k_2 = a' + k_2 = k$.  Thus, it now suffices to lower bound the probability that  $f_{b(k-a)} = 1$ for some $A_i$-forbidden pair $(a,b)$.

We will first compute the probability of the complement.  Recall that for each $b$, we choose to set exactly one coefficient $f_{bk_2} = 1$ where $k_2$ is chosen uniformly at random from $\{-b, -b + 1, \dots , b \}$. For each $b$ there are some forbidden values $k_2$, which we may call $A_i$-forbidden values, such that we must avoid setting $f_{bk_2} = 1$.  For each $b$, let $s_b$ be the number of $A_i$-forbidden values.  Then the probability that all of these forbidden values are avoided is 
\[
\prod_{b \in B} \frac{2b+1 - s_b}{2b+1} \,.
\]
However, note that 
\[
\sum_{b \in B}s_b = \Omega(L^{1.1})
\]
because we are combining over $\Omega(L^{0.6})$ possible values of $a \in A_i$ and each value of $a$ forbids $\Omega(\sqrt{L})$ distinct pairs.  Thus
\[
\prod_{b \in B} \frac{2b+1 - s_b}{2b+1} \leq \prod_{b \in B}\left( 1 - \frac{s_b}{2L}\right) \leq \exp\left(-\frac{1}{L} \cdot \sum_{b \in B} s_b\right) = \exp\left(-\Omega(L^{0.1})\right) \,.
\]
We conclude that for a fixed $k$, with at least $1 - \exp\left(-\Omega(L^{0.1})\right)$ probability (over the random assignment of values to the $f_{ak_1},f_{bk_2}$), there are $a\in A_i, b \in B$ and $k_1,k_2$ such that
\begin{itemize}
    \item $f_{ak_1}f_{bk_2} = 1$
    \item $\abs{\langle ak_1bk_2| Lk \rangle} \geq \frac{1}{\poly(L)}$
\end{itemize}

Union bounding over all of $i = 1,2, \dots , d$ completes the proof of Property~\ref{prop:statement1}. Now union bounding the result of Property~\ref{prop:statement1} over all $k$ with $ 0.02L < k < 0.98L$ or $ -0.02L < k < -0.98L$, we get that there is positive probability that the property holds simultaneously for all such $k$.

We now plug this assignment of values into the matrix $M^{(A,B)}$, and restrict to the columns indexed by $k$ with $ 0.02L < k < 0.98L$ or $ -0.02L < k < -0.98L$.  For each of these columns, there are $\Omega(L^{0.4})$ rows of $M^{(A,B)}$ that have their only nonzero entry in the column indexed by $k$ and such that the value of this nonzero entry has magnitude at least $1/\poly(L)$.  We can now form a square submatrix $M'$ as follows.  For each column, we look at the $\Omega(L^{0.4})$ possible rows that we can pick and since these rows are indexed by disjoint $(a,b)$, we can always pick one of these rows to ensure that the subset of rows selected so far remains $O(L^{0.6})$-balanced.  At the end, we have constructed a square submatrix $M'$ with the desired properties, completing the proof.
\end{proof}

We can now take the submatrix $M'$ found in Lemma \ref{lemma:good-assignment} and understand its determinant when written as a polynomial in the variables $\{f_{ak_1} \}, \{f_{bk_2}\}$.  We will use the following notation:
\begin{itemize}
    \item Let $m$ be the size of the submatrix $M'$ (so $M'$ is an $m \times m$ matrix and $m \geq 0.95(2L+1)$.
    \item  For each $a \in A, b \in B$ let $t_a$ (respectively $t_b$) be the unique integer in the interval $[-a, a]$ for which the assignment computed in Lemma \ref{lemma:good-assignment} sets $f_{at_a} = 1$.
\end{itemize}

The following result is a simple corollary of Lemma \ref{lemma:good-assignment}.

\begin{corollary}\label{corollary:large-coeff}
Let $P(\{f_{ak_1} \}, \{f_{bk_2}\})$ be the polynomial obtained when writing the determinant of $M'$ as a polynomial in $\{f_{ak_1} \}, \{f_{bk_2}\}$.  Then $P$ is a homogeneous polynomial of degree $2m$.  Furthermore there are nonnegative integers $g_a$ for each $a \in A$ and $g_b$ for each $b \in B$ such that 
\begin{itemize}
    \item $g_a, g_b \leq O(L^{0.6})$ for all $a,b$
    \item $\sum_{a \in A} g_a + \sum_{b \in B} g_b = 2m$
    \item The coefficient of 
    \[
    \prod_{a \in A} f_{at_a}^{g_a}\prod_{b \in B} f_{bt_b}^{g_b}
    \]
    in $P$ has magnitude at least $1/L^{CL}$ for some universal constant $C$.
    \item $P$ has degree $g_a$ when viewed as a polynomial in only $f_{at_a}$ for each $a \in A$ and $P$ has degree $g_b$ when viewed as a polynomial in only $f_{bt_b}$ for each $b \in B$ 
\end{itemize}
\end{corollary}
\begin{proof}
Note that the rows of $M'$ may be indexed by pairs of integers $(a,b)$ with $a \in A$ and $b \in B$.  For each $a \in A$, $g_a$ is simply the number of rows of $M'$ that have $a$ as the first coordinate.  Similarly, $g_b$ is the number of rows of $M'$ that have $b$ as the second coordinate.  The balancedness condition in Lemma \ref{lemma:good-assignment} gives us that $g_a, g_b \leq O(L^{0.6})$ for all $a,b$.  We also immediately get the second condition that  
\[
\sum_{a \in A} g_a + \sum_{b \in B} g_b = 2m \,.
\]
Next, there is only one way to get the term $\prod_{a \in A} f_{at_a}^{g_a}\prod_{b \in B} f_{bt_b}^{g_b}$ in the expansion of the determinant of $M'$.  The coefficient is exactly the value of the determinant when we set $f_{at_a} = 1, f_{bt_b} =  1 $ for all $a \in A, b \in B$ and set all of the other variables to $0$.  Lemma \ref{lemma:good-assignment} tells us that the value of this determinant has magnitude at least $1/\poly(L)^L$.  This gives us the third of the desired conditions.

Finally, to verify the fourth condition, note that the maximum possible degree of any monomial of $P$ in the variable $f_{at_a}$ is equal to the number of rows that contain the variable $f_{at_a}$ somewhere.  The number of such rows is exactly $g_a$.  Similarly, we get the same result for $f_{bt_b}$  for all $b \in B$ and this completes the proof.
\end{proof}

Combining Corollary \ref{corollary:large-coeff} and the anticoncentration bound in Corollary \ref{coro:multivariate-anticoncentration}, we can prove the main result of this section.
\begin{lemma}\label{lem:almostfullrank}
Assume that the values of the spherical harmonic coefficients $f_{ak_1}, f_{bk_2}$ are $\delta$-smoothed.  There is a square submatrix $M'$ of $M^{(A,B)}$ of size at least $0.95(2L+1) \times 0.95(2L+1)$ such that the following holds. There is an absolute constant $K$ such that for any $\eps < (\delta/L)^K$, with probability at least 
\[
1 - \eps^{\Omega(L^{0.4})} \,,
\] 
the $0.0001L$ \ts{th} smallest singular value of $M'$ is at least $\eps$.
\end{lemma}
\begin{proof}
$M'$ will be the matrix constructed in Lemma \ref{lemma:good-assignment}.  First, smooth all of the values of $f_{ak_1}, f_{bk_2}$ for $k_1 \neq t_a, k_2 \neq t_b$ for all $a \in A, b \in B$.  Now consider the polynomial $P(\{f_{at_a}\}, \{f_{bt_b} \})$ obtained by writing the determinant of the matrix $M'$ as a polynomial in variables $ \{f_{at_a}\}, \{f_{bt_b} \}$ (and plugging in values for all of the other $f_{ak_1}, f_{bk_2}$).  We can now use Corollary \ref{corollary:large-coeff} and apply Corollary \ref{coro:multivariate-anticoncentration} to this polynomial $P$ with variables $\{f_{at_a}\}, \{f_{bt_b} \}$ to deduce that the magnitude of the determinant is at least
\[
\frac{1}{L^{CL}}\delta^{2m} \left(\frac{\eps'}{e}\right)^{80m}
\]
(for some universal constant $C$) with probability 
\[
1 - \eps'^{\Omega(L^{0.4})} \,.
\]
where $\eps'$ is a parameter that will be set later.  Note that to obtain the probability bound, we used the fact that the degree of $P$ in each individual variable is $O(L^{0.6})$.

Note that $0.95(2L+1) \leq m \leq  2L+1$, and also that the largest singular value of $M'$ is at most $(10L)^2$ (since all of the $f_{ak_1}, f_{ak_2}$ have magnitude at most $1$ and also all of the CG-coefficients have magnitude at most $1$).  Also, the product of the singular values of $M'$ is equal to the magnitude of the determinant so with $1 - \eps'^{\Omega(L^{0.4})}$ probability, the $0.0001L$ \ts{th} smallest singular value is at least $(\delta \eps'/L)^{C'}$ for some absolute constant $C'$.  Choosing $\eps'$ so that $\eps = (\delta \eps'/L)^{C'}$, completes the proof.

\end{proof}

\subsubsection{Finishing the Proof of Lemma \ref{lem:well-conditioned}}\label{sec:extra-subspace}

In the previous section, we found a submatrix of $M^{(A,B)}$ that robustly has almost full rank.  In order to show that $M^{\textsf{truth}}$ robustly has full rank, we will add a few more rows corresponding to $(a,b)$ with $a \notin A, b \notin B$.  The analysis in this section will be significantly simpler because we only need to add a small number of rows so we can actually ensure that the random smoothing in the new rows we add is independent.  We formalize this below.

For each odd integer $l_1$ between $0.5L$ and $0.99L$, we consider the row of $M^{\textsf{truth}}$ indexed by $(l_1, l_1 + 1)$, which we denote $M^{\textsf{truth}}_{(l_1,l_1 + 1)}$.  Let $M^{\textsf{aux}}$ be the matrix whose rows are $M^{\textsf{truth}}_{(l_1,l_1 + 1)}$ where $l_1$ ranges over all odd integers in the range $[0.5L, 0.99L]$.  The key lemma that we will show is stated below.
\begin{lemma}\label{lemma:extra-subspace}
Assume that the spherical harmonic coefficients $f_{lk}$ are $\delta$-smoothed.  Let $V \subset \C^{2L+1}$ be a fixed subspace of dimension at most $0.11L$.  There is an absolute constant $K$ such that for any $0< \eps < (\delta/L)^{-K}$, with probability at least $1 - \eps^{0.001L}$ over the smoothing, for all unit vectors $v \in V$, 
\[
\norm{M^{\textsf{aux}}v}_2 \geq \frac{\delta^2 \eps^2}{20L} \,.
\]
\end{lemma}
\begin{proof}
Let $v$ be a unit vector in $\C^{2L+1}$.  Let $r_{l_1}$ be the row of $M^{\textsf{aux}}$ indexed by $(l_1, l_1 + 1)$ i.e. 
\[
r_{l_1} = M^{\textsf{truth}}_{(l_1,l_1 + 1)} \,.
\]
Then $\langle r_{l_1}, v \rangle$ is a homogeneous degree-$2$ polynomial in the variables  $f_{l_1k_1}, f_{(l_1+1) k_2}$ that is linear in each of the individual variables.  We claim that some coefficient of this polynomial has magnitude at least $1/(10L)$.  To see this, note that some entry of $v$ must have magnitude at least $1/\sqrt{2L+1}$ (since $v$ is a unit vector). Without loss of generality the entry of $v$ indexed by $k$ has magnitude at least $1/\sqrt{2L+1}$.  Next, by Fact \ref{fact:CG-orthogonality},
\[
\sum_{k_1 + k_2 = k}\langle l_1k_1(l_1 + 1)k_2| Lk \rangle^2 = 1 \,,
\]
so for some $k_1 + k_2 = k$, we have
\[
|\langle l_1k_1(l_1 + 1)k_2| Lk \rangle | \geq 1/\sqrt{2L+1} \,.
\]
This implies that coefficient of $f_{l_1k_1}f_{(l_1 + 1)k_2}$ in $\langle r_{l_1}, v \rangle$ has magnitude at least $1/(10L)$.  Recall that the expression $\langle r_{l_1}, v \rangle$ is multilinear and has total degree $2$ in the variables  $f_{l_1k_1}, f_{(l_1+1) k_2}$.  This means, by Corollary \ref{coro:anti-concentration-general}, that
\[
\Pr \left[ |\langle r_{l_1}, v \rangle| \leq  \frac{\delta^2 \eps^2}{10L} \right] \leq O(\eps) \,.
\]

Next, observe that the randomness in all of the rows of $M^{\textsf{aux}}$ is independent.  Thus, for any fixed vector $v$, 
\[
\Pr \left[ \norm{M^{\textsf{aux}}v}_2 \leq \frac{\delta^2 \eps^2}{10 L} \right] \leq \left( O(\eps)\right)^{0.24L}\,.
\]
Now we can consider an $\gamma$-net, say $\Gamma$, of all unit vectors in $V$ with $\gamma = \frac{\delta^2 \eps^2}{10^4 L^3}$  which has size 
\[
|\Gamma| \leq \left( \frac{10^5 L^3}{\delta^2 \eps^2}\right)^{0.11L} \,.
\]
As long as $\eps < (\delta/L)^{K}$ for some sufficiently large universal constant $K$, a union bound tells us that with probability $1 - \eps^{0.001L}$, 
\[
\norm{M^{\textsf{aux}}v}_2 \geq \frac{\delta^2 \eps^2}{10L}
\]
for all $v \in \Gamma$.  This then implies for all $v \in V$,
\[
\norm{M^{\textsf{aux}}v}_2 \geq \frac{\delta^2 \eps^2}{10L} - \frac{\delta^2 \eps^2}{10^4 L^3} \norm{M^{\textsf{aux}}}_{\textsf{op}} \geq \frac{\delta^2 \eps^2}{20L} 
\]
where we use a trivial upper bound on $\norm{M^{\textsf{aux}}}_{\textsf{op}}$ since its entries are bounded by $O(L)$.  This completes the proof.
\end{proof}

We can now combine Lemma \ref{lemma:extra-subspace} with Lemma \ref{lem:almostfullrank} to prove that the entire matrix $M^{\textsf{truth}}$ is well conditioned, which will complete the proof of Lemma \ref{lem:well-conditioned}.  

\begin{proof}[Proof of Lemma \ref{lem:well-conditioned}]
Let $\eps' \leq  ( \delta/L)^{O(1)}$ be a parameter that will be set later.  Lemma \ref{lem:almostfullrank} implies that with probability $1 - \eps'^{\Omega(L^{0.4})}$, there is a subspace $U \subset \C^{2L+1}$ of dimension at least 
\[
0.95(2L+1) - 0.0001L \geq 1.899L
\]
such that for any $u \in \C^{2L+1}$, if the projection of $u$ onto $U$ has length at least $1$, then $\norm{M^{(A,B)}u}_2 \geq \eps'$.  To see this, it suffices to take the top-$1.899L $-singular subspace of the matrix $M^{(A,B)}$.

Let $V$ be the orthogonal complement of $U$.  Lemma \ref{lemma:extra-subspace} implies that with  $1 - \eps'^{0.001L}$ probability,
\[
\norm{M^{\textsf{aux}}v}_2 \geq \frac{\delta^2 \eps'^2}{20L}
\]
for all unit vectors $v \in V$.  Note that this application of Lemma \ref{lemma:extra-subspace} is valid because the rows of $M^{\textsf{aux}}$ and $M^{(A,B)}$ do not have any overlapping variables, so we can imagine sampling the smoothing in the rows of $M^{(A,B)}$ first, which fixes the subspaces $U,V$, and then applying Lemma \ref{lemma:extra-subspace}.

Next, observe that the largest singular value of $M^{\textsf{aux}}$ is at most $(10L)^2$ (just by using the trivial upper bound on its individual entries).  Now given a unit vector say $z \in \C^{2l+1}$ we can write it as a sum $z_U + z_V$ obtained by projecting $z$ onto $U$ and $V$ respectively.  We consider two cases.  If the projection onto $U$ satisfies
\[
\norm{z_U}_2 \geq \frac{1}{(100L)^2} \frac{\delta^2 \eps'^2}{20L}
\]
then we have 
\[
\norm{M^{(A,B)}z}_2 \geq \frac{1}{(100L)^2} \frac{\delta^2 \eps'^2}{20L} \cdot  \eps' \,.
\]
Otherwise, we have 
\[
\norm{M^{\textsf{aux}}z}_2 \geq \norm{M^{\textsf{aux}}z_V}_2 - \norm{M^{\textsf{aux}}z_U}_2 \geq \frac{\delta^2 \eps'^2}{40L} - (10L)^2 \cdot \frac{1}{(100L)^2} \frac{\delta^2 \eps'^2}{20L} \geq \frac{\delta^2 \eps'^2}{10^3L} \,.
\]
In both cases, 
\[
\norm{M^{\textsf{truth}}z}_2 \geq (\eps' \delta/ L)^{O(1)} \,,
\]
so actually the above is true for all unit vectors $z$.  Thus, ensuring that $\eps \leq (\delta/L)^{K}$ for some sufficiently large constant $K$ and choosing $\eps' = \eps^{\Omega(1)}$ appropriately completes the proof.

\end{proof}

\subsection{Completing the Analysis of {\sc Frequency Marching with Long Stride} }

Using Lemma \ref{lem:well-conditioned}, we can complete the proof of Lemma \ref{lem:alg2-analysis}.

\begin{proof}[Proof of Lemma \ref{lem:alg2-analysis}]

Recall the discussion at the beginning of Section \ref{sec:well-conditioned}.  In {\sc Frequency Marching with Long Stride}, the optimization problem that we need to solve is as follows.  For variables $X = \{x_{L(-L)}, \dots , x_{LL} \}$ we want to solve
\[
\min_X \norm{MX - \mcl{I}}_2^2 \,.
\] 
where $M$ is a $(0.9999L)^2 \times (2L + 1)$ matrix where
\begin{itemize}
    \item The rows of $M$ are indexed by pairs of integers $a,b$ with $a,b \leq 0.9999L$ and the columns are indexed by integers $m = -L, \dots , L$ with entries given by 
    \[
    M_{(a,b) m} = \sum_{\substack{k_1 + k_2 = -m \\ \abs{k_1} \leq a, \abs{k_2} \leq b}} (-1)^{m} \langle ak_1bk_2| L(-m) \rangle \wt{f_{ak_1}} \wt{f_{bk_2}} 
    \]
   where $\wt{f_{lm}}$ (with $l \leq L - 1$) are the estimates for the lower-degree spherical harmonic coefficients. 
    \item The entries of $\mcl{I}$ are indexed by pairs of integers $a,b$ with $a,b \leq 0.9999L$ and are equal to our estimates $\wt{\mcl{I}_{a,b,L}}$ for the invariant polynomial $\mcl{I}_{a,b,L}(f)$
\end{itemize}

By Lemma \ref{lem:well-conditioned}, with probability at least $1 - 2^{-L^{0.4}}$, after the $\delta$-smoothing, the spherical harmonic coefficients of $f$ satisfy that the matrix  $M^{\textsf{truth}}$ (recall Definition \ref{def:key-matrix}) has smallest singular value at least $(\delta/L)^K$ where $K$ is an absolute constant. 

Now, as long as $\gamma' < (\delta/L)^{O(1)}$ is sufficiently small, we have that the smallest singular value of $M$ is at least  $0.5(\delta/L)^K$.  To see this, note that $M$ is obtained from $M^{\textsf{truth}}$ by replacing the values of the true coefficients $f_{lm}$ with the values of our estimates $\wt{f_{lm}}$ so we have 
\begin{equation}\label{eq:estbound1}
\norm{M - M^{\textsf{truth}}}_F^2 \leq \poly(L/\delta) \gamma' \,.
\end{equation}
Also, by Lemma \ref{lem:est-invariant}, we can ensure that with $1 - 2^{-L}$ probability, our estimates $\wt{\mcl{I}_{a,b,L}}$ all satisfy 
\begin{equation}\label{eq:estbound2}
\left \lvert \wt{\mcl{I}_{a,b,L}} - \mcl{I}_{a,b,L}(f) \right \rvert \leq \sqrt{\gamma} \,.
\end{equation}
Now let $f_L = \{f_{L(-L)}. \dots , f_{LL} \}$.  Since by definition,
\[
M^{\textsf{truth}} f_L = \{ \mcl{I}_{a,b,L}(f) \}_{a,b \leq 0.9999L} 
\]
we can now use (\ref{eq:estbound1}) and (\ref{eq:estbound2}) to get that
\[
\norm{Mf_L - \mcl{I}}_2^2 \leq \poly(L/\delta) (\gamma' + \gamma) \,.
\]
Now let $\rho_{\min}$ be the smallest singular value of $M$.  For any $X$,
\[
\norm{ M(f_L - X)}_2 \geq \rho_{\min} \norm{f_l - X}_2 
\]
so the solution that is actually output by the algorithm {\sc Frequency Marching with Long Stride}, say $\wt{f_L} = \{\wt{f_{L(-L)}}. \dots , \wt{f_{LL}} \} $ must satisfy 
\[
\rho_{\min} \norm{f_L - \wt{f_L}}_2  \leq \norm{ Mf_L - M\wt{f_L}}_2 \leq 2\norm{Mf_L - \mcl{I}}_2 \leq \poly(L/\delta) \sqrt{\gamma' + \gamma} \,.
\]
Since we have that $\rho_{\min} \geq 0.5(\delta/L)^K$ for an absolute constant $K$, we get
\[
\norm{f_L - \wt{f_L}}_2^2 \leq \poly(L/\delta) (\gamma' + \gamma) \,,
\]
as desired.  It is clear that the failure probability over all of the steps is at most $ 1 - 2^{-L^{0.1}}$ and that the algorithm runs in polynomial time (since the optimization problem we are solving is just linear regression).
\end{proof}

\section{Multiple Shells}\label{sec:multiple-shells}

We can consider a more general setting of our problem where instead of a function $f:S^2 \rightarrow \C$, we have a function $f$ that is defined on multiple spherical shells of radii $r_1, \dots , r_T$.  Let $B_{r}$ denote the spherical shell of radius $r$ (in $\R^3$).  There is some unknown function $f: B_{r_1} \cup \dots \cup B_{r_T} \rightarrow \C$.

We will assume that the radii of all of the shells are lower and upper bounded by some constant so the measures on these shells differ by at most a constant factor.  Thus, up to a constant factor in the reconstruction guarantee, we may equivalently view $f$ as a $T$-tuple of functions $(f^{(1)}, \dots , f^{(T)})$ each from $S^2 \rightarrow \C$ where a rotation $R \in SO(3)$ acts by rotating all of $f^{(1)}, \dots , f^{(T)}$ simultaneously i.e. 
\[
R(f) = (R(f^{(1)}), \dots , R(f^{(T)})) \,.
\]
Our reconstruction goal will be to output a function $\wt{f} = (\wt{f^{(1)}}, \dots , \wt{f^{(T)}})$ such that the distance $d_{SO(3)}(\wt{f}, f)$ is small where we define 
\[
d_{SO(3)}(\wt{f}, f)  = \min_{R \in SO(3)} \norm{\wt{f} - R(f)}_2^2 = \min_{R \in SO(3)} \left( \sum_{j \in [T]}\norm{\wt{f^{(j)}} - R(f^{(j)}) }_2^2 \right) \,.
\]

We assume that the expansion of each of $(f^{(1)}, \dots , f^{(T)})$ into spherical harmonics has degree at most $N$.  We can now formulate the same problem as in Section \ref{sec:problem-formulation}.  Analogous to (\ref{eq:sample-model}), each of our observations is of the form
\begin{equation}\label{eq:multipleshells-sample-model}
\wh{f_{lm}^{j}} = R(f^{j})_{lm} + \zeta 
\end{equation}
for all $l \leq N, -l \leq m \leq l, j \in [T]$ where $\zeta$ has real and imaginary part drawn independently from $N(0, \sigma^2)$.  

We will also assume that all of the coefficients $f_{lm}^{j}$ are $\delta$-smoothed (independently).  Our main theorem for reconstructing a function defined  multiple shells is as follows:

\begin{theorem}\label{thm:multiple-shells}
Let $f = (f^{(1)}, \dots , f^{(T)})$ be a function where each $f^{(j)}: S^2 \rightarrow \C$ is a function whose expansion in spherical harmonics has degree at most $N$.  Also assume  $\norm{f^{(j)}}_2 \leq 1$ for all $j$.  Then given $Q$ observations from (\ref{eq:multipleshells-sample-model}), where
\[
Q = \left( \frac{N}{\delta} \right)^{O(\log N)} \poly\left( T, \sigma, \frac{1}{\eps}\right) 
\]
there is an algorithm (Algorithm \ref{alg:full-multiple-shells}) that runs in $poly(Q)$ time and with probability $0.9$ outputs a function $\wt{f} = (\wt{f^{(1)}}, \dots , \wt{f^{(T)}})$ where each $\wt{f^{(j)}}$ has spherical harmonic expansion of degree at most $N$ and such that
\[
d_{SO(3)}(f, \wt{f}) \leq \eps \,.
\]
\end{theorem}

Throughout this section, we will assume that $N \geq 1$ because the case where $N = 0$ is trivial.  Theorem \ref{thm:multiple-shells} can be proven using almost the same methods as Theorem \ref{thm:main-oneshell}.  First, we attempt to recover the constant-degree coefficients of all of the functions $f^{(1)}, \dots , f^{(T)}$.  Individually, this is simple by Lemma \ref{lem:alg1-analysis} but we then need to align them.  To do this, we learn groups of functions e.g. $(f^{(1)}, f^{(2)}, f^{(3)})$, simultaneously using the same technique as the earlier algorithm {\sc Learn Constant-Degree Coefficients}.  Then we can patch together two groups, say $(f^{(1)}, f^{(2)}, f^{(3)})$ and $(f^{(2)}, f^{(3)}, f^{(4)})$ by finding a rotation that aligns the common parts (here that would be $f^{(2)}, f^{(3)}$).  This will let us align the constant-degree coefficients for all of the functions.  We then run an algorithm that is very similar to  {\sc Frequency Marching with Long Stride} to recover the higher-degree coefficients for each of the functions $f^{(1)}, \dots , f^{(T)}$.  Note that while we could run $T$ independent instances of {\sc Frequency Marching with Long Stride}, we would then need to pay a $T^{O(\log N)}$ term in the sample complexity because we would only be able to guarantee that the condition numbers of the matrices are $\poly(NT/\delta)$ (instead of $\poly(N/\delta)$).

\subsection{Invariant Polynomials for Multiple Shells}

It turns out that with multiple-shells, there is an explicit formula for the degree-$3$ invariant polynomials that maintains the crucial layered structure of the formula in Theorem \ref{thm:explicit-invariants}.  This result is also from \cite{bandeira2018estimation}.
\begin{theorem}[\cite{bandeira2018estimation}]\label{thm:multiple-shell-invariants}
For a function $f = (f^{(1)}, \dots , f^{(T)})$ over multiple shells, the degree-$3$ invariant polynomials are (up to scaling) 
\begin{equation}
\mathcal{I}_{s_1,l_1,s_2, l_2, s_3,l_3}(f) = \sum_{\substack{k_1 + k_2 + k_3 = 0 \\ \abs{k_i} \leq l_i}} (-1)^{k_3} \langle l_1k_1l_2k_2| l_3(-k_3) \rangle f_{l_1k_1}^{(s_1)} f_{l_2k_2}^{(s_2)}f_{l_3k_3}^{(s_3)}
\end{equation}
where $\langle l_1k_1l_2k_2| l_3(-k_3) \rangle$ denotes the Clebsch-Gordan (CG) coefficient.
\end{theorem}

Also note that we can view $f$ as a vector with dimensionality $T(N+1)^2$ given by the coefficients in the spherical harmonic expansions of its components $( \{ f_{lm}^{(1)}\}_{l \leq N}, \dots ,  \{ f_{lm}^{(T)}\}_{l \leq N})$ and that any $R \in SO(3)$ acts linearly on this vector and preserves its $L^2$ norm.  Thus, we can apply Lemma \ref{lem:est-invariant} to estimate the invariant polynomials in the multiple-shell setting as well.

\subsection{Algorithm for Multiple Shells}

We now summarize our algorithm for multiple shells.  First, we have an algorithm for recovering constant-degree coefficients for a function with $T$ shells, but the runtime and sample complexity may depend exponentially (or worse) on $T$ so we will only use this algorithm for $T \leq 5$.  The algorithm is very similar to the earlier algorithm {\sc Learn Constant-Degree Coefficients}.

Note that we can view $f$ as a vector with dimensionality $T(N+1)^2$ given by the coefficients in the spherical harmonic expansions of its components $( \{ f_{lm}^{(1)}\}_{l \leq N}, \dots ,  \{ f_{lm}^{(T)}\}_{l \leq N})$ and that any $R \in SO(3)$ acts linearly on this vector and preserves its $L^2$ norm.  Thus, we can apply Lemma \ref{lem:est-invariant} to estimate the invariant polynomials in the multiple-shell setting as well. 

Below, we will use $C$ to denote a sufficiently large universal constant.
\begin{algorithm}[H]\label{alg:const-degree-multiple-shells}
\caption{{\sc Learn Constant-degree Coefficients for Multiple Shells} }
\begin{algorithmic} 
\State \textbf{Input:} Parameters $\sigma, \gamma, T \leq 5$ 
\State \textbf{Input:} $Q = \poly(\sigma,  1/\gamma)$ samples of $f$ from the sampling model (\ref{eq:multipleshells-sample-model})
\State Compute invariant polynomials $P_1, \dots , P_k$ that generate 
\[
\C^{SO(3)}[ x^{(1)}_{00}, \dots , x^{(1)}_{CC}, \dots , x^{(T)}_{00}, \dots , x^{(T)}_{CC}] \,,
\]
the invariant ring for coefficients of spherical harmonics of degree at most $C$

\State Obtain estimates $\wt{P_1}, \dots , \wt{P_k}$ for $P_1\left(\{f_{lm}^{(j)} \}_{l \leq C, j \in [T]} \right), \dots , P_k\left(\{f_{lm}^{(j)} \}_{l \leq C, j \in [T]} \right)$ using Lemma \ref{lem:est-invariant}
\State Let $K$ be a sufficiently large universal constant (in terms of $C, P_1, \dots , P_k$)
\State Grid search for values of $\{f_{lm}^{(j)}\}_{l \leq C, j \in [T]}$ with $|f_{lm}^{(j)}| \leq 1$ with discretization $(1/\gamma)^{O(K)}$
\For{each guess $\{\wt{f_{lm}^{(j)}}\}_{l \leq C, j \in [T]}$}
\State Check if for all $j \in [k]$
\[
\left \lvert P_j\left( \{\wt{f_{lm}^{(j)}}\}_{l \leq C, j \in [T]} \right) - \wt{P_j} \right \rvert \leq 0.2 (1/\gamma)^{K}
\]
\If{above check passes}
\State \textbf{Output:} $\{\wt{f_{lm}^{(j)}}\}_{l \leq C, j \in [T]}$
\State \textbf{break}
\EndIf
\EndFor
\end{algorithmic}
\end{algorithm}

Next, we describe our algorithm for iteratively recovering the higher-degree spherical harmonic coefficients via frequency marching.
\begin{algorithm}[H]
\caption{{\sc Frequency Marching for Multiple Shells} }
\begin{algorithmic} 
\State \textbf{Input:} Parameters $\delta, \sigma, \gamma$
\State \textbf{Input:} Indices $L \in [N], j \in [T]$ and estimates $\{ \wt{f_{lm}^{(1)}} \}_{l \leq L - 1}$ of coefficients of degree less than $L$ of $f^{(1)}$
\State \textbf{Input:} $Q = \poly(L T \sigma/(\delta \gamma))$ samples of $f$ from the sampling model (\ref{eq:sample-model})

\State 
\For{all integers $a,b \leq 0.9999 L$}
\State Compute estimate $\wt{\mcl{I}_{1,a,1,b,j,L}}$ of $\mcl{I}_{1,a,1,b,j,L}(f)$ using Lemma \ref{lem:est-invariant} 
\EndFor
\For{all integers $a,b \leq 0.9999 L$}
\State For variables $X = \{ x_{L(-L)}, \dots , x_{LL} \}$ , define the vector $M_{(a,b)}$ such that
\[
  M_{(a,b)} \cdot  X = \sum_{\substack{k_1 + k_2 + k_3 = 0 \\ \abs{k_1} \leq a, \abs{k_2} \leq b, \abs{k_3} \leq L}} (-1)^{k_3} \langle ak_1bk_2| L(-k_3) \rangle \wt{f_{ak_1}^{(1)}} \wt{f_{bk_2}^{(1)}} x_{Lk_3}
\]
\EndFor
\State Let $M$ be the $(0.9999L)^2 \times (2L + 1)$ matrix with rows given by $M_{(a,b)}$ for $a,b \leq 0.9999 L$
\State Let $\mcl{I}$ be the vector of length $(0.9999L)^2$ with entries $\wt{\mcl{I}_{1,a,1,b,j,L}}$ for $a,b \leq 0.9999 L$
\State \textbf{Output:} $\{\wt{f_{L(-L)}^{(j)}}, \dots , \wt{f_{LL}^{(j)}} \}$  as the solution to
\[
 \arg \min_{X} \left( \norm{MX - \mcl{I}}_2^2\right) \,.
\]
\end{algorithmic}
\end{algorithm}
Note that this algorithm is exactly the same as the algorithm in the one-shell case except we only use estimates for the coefficients of $f^{(1)}$ to set up our linear system and recover the coefficients of $f^{(j)}$ for any $j$.  This exploits the structure of the invariant polynomials (see Theorem \ref{thm:multiple-shell-invariants}) where we set $s_1 = s_2 = 1, s_3 = j$.

Now we describe our full algorithm for recovering the remaining coefficients and aligning all of the functions.  We essentially run the above algorithm over $5$-tuples of indices that share $4$ common indices and use the $4$ common indices to align them.  This allows us to align all of the functions one-by-one.  We then use {\sc Frequency Marching for Multiple Shells} to recover the higher-degree coefficients for each shell.
\begin{algorithm}[H]
\caption{{\sc $SO(3)$ Reconstruction with Multiple Shells} }
\begin{algorithmic} 
\State \textbf{Input:} Parameters $N,T,  \delta, \sigma , \eps$ 
\State \textbf{Input:} $Q$ samples of $f$ from the sampling model (\ref{eq:multipleshells-sample-model}) where 
\begin{itemize}
    \item $f = (f^{(1)}, \dots , f^{(T)})$ is an unknown function whose spherical harmonic expansions have degree at most $N$
    \item $\norm{f}_2 \leq 1$
    \item The coefficients of $f$ are $\delta$-smoothed
    \item The number of samples is
    \[
    Q = \left( \frac{N}{\delta}\right)^{O(\log N)} \poly\left(T, \sigma, \frac{1}{\eps} \right) \,.
    \]
\end{itemize}
\State Set $\gamma  = \left( \delta/N\right)^{O(\log N)} (\eps/T)^{O(1)}$
\State Run {\sc Learn Constant-degree Coefficients for Multiple Shells} with parameters $(\sigma, \gamma, T = 5)$ on 
\[
\left( f^{(1)}, f^{(2)},f^{(3)},f^{(4)} , f^{(5)} \right)
\]
to obtain estimates
\[
\{ \wt{f_{lm}^{(j)}} \}_{l \leq C, j \leq 5} \,. 
\]

\For{all $6 \leq i \leq T$} 
\State Run {\sc Learn Constant-degree Coefficients for Multiple Shells} with parameters $(\sigma, \gamma, T = 5)$ on 
\[
\left( f^{(1)}, f^{(2)},f^{(3)},f^{(4)} , f^{(i)} \right)
\]
\indent to obtain estimates 
\[
\{ \wh{f_{lm}^{(j)}} \}_{l \leq C, j  \in \{1,2,3,4, i \}} \,.
\]

\State Find $R \in SO(3)$ that minimizes 
\[
\sum_{l \leq C}\sum_{m = -l}^l \sum_{j \in \{1,2,3,4 \}} \left \lvert R(\wh{f^{(j)}})_{lm} - \wt{f_{lm}^{(j)}}\right \rvert^2 \,.
\]
\State Set $\{ \wt{f_{lm}^{(i)}} \}_{l \leq C} = \{ R(\wh{f^{(i)}})_{lm} \}_{l \leq C}$
\EndFor

\For{j = 1,2, \dots , T}
\For{L = C+1 , \dots , N}
\State Run {\sc Frequency Marching for Multiple Shells} with parameters $\delta, \sigma, \gamma$, indices $L, j$ and estimates $\{\wt{f_{lm}^{(1)}} \}_{l \leq L-1} $ to obtain solution $\{\wt{f_{L(-L)}^{(j)}}, \dots , \wt{f_{LL}^{(j)}} \}$ 
\EndFor
\EndFor

\State \textbf{Output:} $\wt{f^{(j)}} = \{\wt{f_{lm}^{(j)}} \}_{l \leq N}$ for all $j = 1,2, \dots , T$
\end{algorithmic}
\label{alg:full-multiple-shells}
\end{algorithm}
\begin{remark}
To solve the minimization over $R \in SO(3)$, we will simply grid search over a sufficiently fine grid.
\end{remark}

The following two lemmas correspond to Lemma \ref{lem:alg1-analysis} and Lemma \ref{lem:alg2-analysis}.  The first, about {\sc Learn Constant-Degree Coefficients For Multiple Shells} follows from exactly the same proof as Lemma \ref{lem:alg1-analysis}. 

\begin{lemma}\label{lem:const-degree-multiple-shells}
Let $1 \leq T \leq 5$ be an integer.  With probability $1 - 2^{-1/\gamma}$ over the samples, the output of {\sc Learn Constant-Degree Coefficients for Multiple Shells} for a function with $T$ shells satisfies the property that there is a rotation $R \in SO(3)$ such that 
\[
\sum_{l = 0}^C \sum_{m = -l}^l \sum_{j \in [T] }\left\lvert \wt{f_{lm}^{(j)}} - R(f^{(j)})_{lm}\right \rvert^2 \leq \gamma \,.
\]
This statement holds with no smoothing on the coefficients of $f$.
\end{lemma}

The next lemma follows from exactly the same proof as Lemma \ref{lem:alg2-analysis}.
\begin{lemma}\label{lem:iterate-multiple-shells}
Fix an index $L$.  With probability $1 - 2^{-L^{0.1}}$ over the random smoothing of the coefficients, the following holds.  Given initial estimates for  {\sc Frequency Marching for Multiple Shells} that satisfy 
\[
\sum_{l = 0}^{0.9999L} \sum_{m = -l}^l |\wt{f_{lm}^{(1)}} - f_{lm}^{(1)}|^2 \leq \gamma' 
\]
for some sufficiently small $\gamma' \leq (\delta / L)^{O(1)}$ and any index $j$, the algorithm, with probability $1 - 2^{-LT}$ over the samples, outputs a solution $\{\wt{f_{L(-L)}^{(j)}}, \dots , \wt{f_{(LL)}^{(j)}} \}$ satisfying
\[
\sum_{m = -L}^L |\wt{f_{Lm}^{(j)}} - f_{Lm}^{(j)}|^2 \leq \poly(L/\delta) (\gamma' + \gamma) \,.
\]
\end{lemma}

Before we can prove Theorem \ref{thm:multiple-shells}, we need one more step.  One potential difficulty that could arise when attempting to align is that some of the functions are invariant under nontrivial rotations.  This would then prevent us from uniquely patching together the groups.  We will show in the proceeding section that this actually does not happen due to the random smoothing.  We complete the proof of Theorem \ref{thm:multiple-shells} afterwards.

\subsection{Smoothing Prevents Rotation Invariance}

\begin{claim}\label{claim:single-rotation}
Consider a function with $f$ whose degree $1$ spherical harmonic coefficients are $f_1 = \{ f_{1(-1)}, f_{10}, f_{11} \}$.  Let $R \in SO(3)$ be a fixed rotation such that $\norm{R - I}_F \geq \tau $ (where we use the $3 \times 3$ matrix corresponding to $R$) for some parameter $\tau$.  Consider $\delta$-smoothing the coefficients of $f$.  Then for any parameter $\rho$, 
\[
\Pr [ \norm{R(f_1)  - f_1}_2^2 \leq  \tau^2 \rho^2\delta^2 ] \leq O(\rho)\,. 
\]
where the probability is over the random smoothing. 
\end{claim}
\begin{proof}
Note that viewing $f_1$ as a vector in $\C^3$, we can write 
\[
R(f_1) - f_1 = (M_R - I)f_1
\]
where $M_R$ is the matrix corresponding to the linear map given by $R$ on the spherical harmonic coefficients of degree $1$.  Note that by assumption, there must exist a unit vector $v \in \R^3$ such that $\norm{Rv - v}_2 \geq \Omega(\tau)$.  Now consider the function from $h: S^2 \rightarrow \C$ given by $h(x) = v \cdot x$.  Note that the expansion of $h$ in spherical harmonics has degree exactly $1$ and is equal to say, $\{ h_{1(-1)}, h_{10}, h_{11} \}$ where 
\[
|h_{1(-1)}|^2 + |h_{10}|^2 + |h_{11}|^2 = O(1) \,.
\]
Now we know that 
\[
\norm{R(h) - h}_2 \geq \Omega(\tau)
\]
so therefore, by orthonormality (Fact \ref{fact:orthonormality}), we must have
\[
\norm{M_R - I}_{F} \geq  \Omega(\tau) \,.
\]
Let the real and imaginary parts of $\{ f_{1(-1)}, f_{10}, f_{11} \}$ be $a_{-1}, b_{-1}, a_0, b_0 , a_1, b_1$ respectively.  Note that the expression $\norm{(M_R - I)f_1}_2^2$ is a quadratic polynomial in $a_{-1}, b_{-1}, a_0, b_0 , a_1, b_1$ and the sum of the coefficients of the terms $a_{-1}^2 , \dots , b_1^2$ is exactly $\norm{M_R - I}_{F}^2$.  Thus, without loss of generality, the coefficient of  $a_{-1}^2$ has magnitude at least $\Omega(\tau^2)$.  Now we can sample the smoothing in all of the other variables $b_{-1}, a_0, b_0 , a_1, b_1$ first and then use Claim \ref{claim:univariate-poly} to get the desired inequality.
\end{proof}

\begin{claim}\label{claim:multipleshell-rotation}
Consider a function $f = (f^{(1)}, \dots , f^{(T)})$ where $T \geq 4$ and assume $\norm{f}_2 \leq 1$.  Also, assume that the spherical harmonic coefficients of each of  $f^{(1)}, \dots , f^{(T)}$ are $\delta$-smoothed.  Let $0 < \tau < 0.5$ be a parameter and let $\rho$ be a parameter with $\rho \leq c (\delta \tau)^{10}$ for some sufficiently small absolute constant $c$.  Then with probability at least $1 - \rho$ over the random smoothing,
\[
\norm{R(f) - f}_2^2 \geq \tau^2 \rho^2 \delta^2
\]
for all for all $ R \in SO(3)$ with $\norm{R - I}_F \geq \tau$.
\end{claim}
\begin{proof}
It suffices to prove the desired inequality in the case that each function $f^{(j)}$ has expansion in spherical harmonics of degree exactly $1$ (since $SO(3)$ acts essentially independently on spherical harmonics of each degree, in the general case we can simply restrict to the degree-$1$ part).

Let $\mcl{S}$ be an $\eps$-net of matrices $R \in SO(3)$ (in Frobenius norm) where $\eps = c' \tau \rho \delta$ for some sufficiently small constant $c'$.  Note that we can ensure $|\mcl{S}| \leq ( 10 /(c'\tau \rho \delta))^{3} $.  Then by Claim \ref{claim:single-rotation} and a union bound, with probability
\[
1 - (O(\rho))^T \cdot ( 10 /(c'\tau \rho \delta))^{3} \geq 1 - \rho
\]
all rotations in $\mcl{S}$ such that $\norm{R - I}_F \geq \tau/2$ satisfy
\[
\norm{R(f) - f}_2^2 \geq \frac{\tau^2 \rho^2 \delta^2}{4} \,.
\]
Now since  $\mcl{S}$ is an $\eps$-net where $\eps = c' \tau \rho \delta$, for any $R \in SO(3)$ with $\norm{R - I}_F \geq \tau$, we can find $R' \in \mcl{S}$ with $\norm{R' - I}_F \geq \tau/2$ and $\norm{R' - R}_F \leq \eps$.  For $R \in SO(3)$, let $M_R$ denote the matrix in $\C^{3 \times 3}$ corresponding to the linear map defined by $R$ on the spherical harmonic coefficients of degree $1$.  It can be verified that the map $R \rightarrow M_R$ is $C$-Lipchitz (with respect to Frobenius norm) for some absolute constant $C$.

We now have
\[
\norm{R(f) - f}_2 \leq \norm{R'(f) - f}_2 + \norm{R'(f) - R(f)}_2 \leq \frac{\tau \rho \delta}{2} + \norm{M_R - M_R'}_F \cdot  \norm{f}_2 \leq \tau \rho \delta
\]
since we can choose $c'$ sufficiently small.  This completes the proof.   
\end{proof}

\subsection{Proof of Theorem \ref{thm:multiple-shells}}

Now we are ready to complete the proof of Theorem \ref{thm:multiple-shells}.
\begin{proof}[Proof of Theorem \ref{thm:multiple-shells}]
If $T \leq 5$ then there is no alignment necessary.  Now consider when $T > 5$.  First, we apply Claim \ref{claim:multipleshell-rotation} on $\left(f^{(1)}, f^{(2)},f^{(3)},f^{(4)} \right)$ with $\rho = 0.01$ to deduce that with $0.99$ probability, there are no rotations that are far from identity that almost preserve all of $\left(f^{(1)}, f^{(2)},f^{(3)},f^{(4)} \right)$ simultaneously.

Now by Lemma \ref{lem:const-degree-multiple-shells}, we can assume that our estimates $\{ \wt{f_{lm}^{(j)}} \}_{l \leq C, j \leq 5}$ satisfy 
\[
\sum_{l = 0}^C \sum_{m = -l}^l \sum_{j =1 }^5 \left\lvert \wt{f_{lm}^{(j)}} - f^{(j)}_{lm}\right \rvert^2 \leq \gamma \,.
\]
While the lemma is stated up to rotation, we can assume without loss of generality that the rotation is identity since otherwise, we can simply pretend that the unknown function $f$ is actually that rotation applied to $f$.

Also, for $i \geq 6$, we can assume that our estimates 
$\{ \wh{f_{lm}^{(j)}} \}_{l \leq C, j \in \{1,2,3,4 , i \}}$ satisfy 
\[
\sum_{l = 0}^C \sum_{m = -l}^l \sum_{j \in \{1,2,3,4,i \}} \left\lvert \wh{f_{lm}^{(j)}} - R_0(f^{(j)})_{lm}\right \rvert^2 \leq \gamma 
\]
for some rotation $R_0$.  Thus, we have 
\[
\sum_{l \leq C}\sum_{m = -l}^l \sum_{j \in \{1,2,3,4 \}} \left \lvert R_0^{-1}(\wh{f^{(j)}})_{lm} - \wt{f_{lm}^{(j)}}\right \rvert^2  \leq O(\gamma) \,.
\]
Now assume that the actual rotation computed by the algorithm is $R = R' R_0^{-1}$.  By Claim \ref{claim:multipleshell-rotation}, we must have $\norm{R' - I }_F \leq \poly(1/\delta)\gamma$. For a rotation $R$, let $M_R$ be the matrix defining the associated linear map on spherical harmonics of degree at most $C$.  Since $C$ is an absolute constant, this map is $C'$-Lipchitz for some absolute constant $C'$ (with respect to Frobenius norm).  Thus, we have
\begin{align*}
\sum_{l = 0}^C \sum_{m = -l}^l  \left\lvert R(\wh{f^{(i)}})_{lm} - f^{(i)}_{lm}\right \rvert^2 \leq 2\sum_{l = 0}^C \sum_{m = -l}^l  \left\lvert R_0^{-1}(\wh{f^{(i)}})_{lm} - f^{(i)}_{lm}\right \rvert^2 \\ + 2\sum_{l = 0}^C \sum_{m = -l}^l  \left\lvert R(\wh{f^{(i)}})_{lm} - R_0^{-1}(\wh{f^{(i)}})_{lm}\right \rvert^2 \leq \poly(1/\delta) \gamma \,.
\end{align*}
Overall, we have shown that all of our estimate $\{f^{(j)}_{lm} \}_{j \leq C}$ are close to the truth.  
The remainder of the proof follows from Lemma \ref{lem:iterate-multiple-shells} and repeating the same argument as the proof in Theorem \ref{thm:main-oneshell}.
\end{proof}

\section{Heterogeneous Mixtures}\label{sec:heterogeneous}

We can generalize our results even further to when the sampling model is heterogeneous.  More formally, assume there are $k$ distinct functions $f^{[1]}, \dots , f^{[k]}$ with mixing weights $w_1, \dots , w_k$ (where $w_i \geq 0$ and $w_1 + \dots + w_k = 1$).  Now our observations are obtained as follows: sample $j \in [k]$ according to the distribution $\{w_1, \dots , w_k \}$ and then observe
\begin{equation}\label{eq:sample-heterogeneous}
\wh{f} = R(f^{[j]}) + \zeta \,.
\end{equation}
\begin{remark}
Our observations will still be viewed as a vector of spherical harmonic coefficients as in  (\ref{eq:multipleshells-sample-model}).
\end{remark}

Our main theorem for the heterogeneous case is stated below.
\begin{theorem}\label{thm:heterogeneous}
Let $f^{[1]} ,  \dots , f^{[k]}$  all be functions with $T$ spherical shells and such that the expansion in spherical harmonics on each shell has degree at most $N$.  Also assume  $\norm{f^{[j]}}_2 \leq 1$ for all $j \in [k]$.  Let $w_1, \dots , w_k \geq w_{\min}$ be mixing weights summing to $1$.  Then given $Q$ observations from (\ref{eq:sample-heterogeneous}), where
\[
Q = \left( \frac{N}{\delta} \right)^{k O(\log N)} \poly\left( (T\sigma/(\eps w_{\min}))^k\right) 
\]
there is an algorithm that runs in $poly(Q)$ time and with probability $0.9$ outputs functions  $\wt{f^{[1]}}, \dots , \wt{f^{[k]}}$  and weights $\wt{w_1}, \dots , \wt{w_k}$ such that there is a permutation $\pi$ on $[k]$ such that for all $j \in [k]$
\[
d_{SO(3)}\left(\wt{f^{[j]}}, f^{[\pi(j)]} \right) + |\wt{w_j} - w_{\pi(j)} | \leq \eps \,.
\]
\end{theorem}

To prove Theorem \ref{thm:heterogeneous}, note that the algorithm in Theorem \ref{thm:multiple-shells} is a statistical query algorithm.  In other words, the algorithm does not need to actually work with the samples but instead only works with the values of invariant polynomials $P_1(f), \dots , P_n(f)$ where $n = \poly(NT)$ that it uses samples to estimate.  Thus, it suffices to compute the weights $w_1, \dots , w_k$ and the values $P_1(f^{[j]}), \dots , P_n(f^{[j]})$ for all $j \in [k]$.  It suffices to estimate these to accuracy \[
\gamma  = \left( \delta/N\right)^{O(\log N)} (\eps/T)^{O(1)}
\]
and then we will be done by the argument in the proof of Theorem \ref{thm:multiple-shells}.

Note that there is an absolute constant $C$ such that all of $P_1, \dots , P_n$ satisfy the following properties:
\begin{itemize}
    \item The degree is at most $C$
    \item All coefficients have magnitude at most $ C$
    \item There is some leading monomial i.e. a monomial with degree equal to $\deg(P_i)$ whose coefficient is at least $1/(C\poly(N))$
\end{itemize}
To see this, note that the above statements are clearly true for the degree-$3$ invariant polynomials that we use in our iterative procedures.  The only other polynomials that we use are fixed, independent of the problem parameters.

\subsection{Estimating Invariant Polynomials for Mixtures}

We have the following generalizations of Lemma \ref{lem:est-moments} for estimating invariant polynomials when our observations come from a heterogeneous mixture.  It is also an immediate consequence of the results in \cite{bandeira2018estimation}.

\begin{lemma}[See Section 7.1 in \cite{bandeira2018estimation}]\label{lem:est-mixture-moments}
Let $G$ be a compact group acting linearly on a vector space $V = \C^n$.  Let $x_1, \dots , x_k \in V$ and assume $\norm{g \cdot x}_2 \leq 1$ for all $g \in G$.  Let $w_1, \dots , w_k$ be mixing weights summing to $1$.  Assume we are given $Q$ independent observations $y_1, \dots , y_Q$ of the form
\[
y_j = g_j \cdot x_l + N(0, \sigma^2 I) + iN(0, \sigma^2 I)
\]  
where $l$ is sampled from $[k]$ according to $\{w_1, \dots , w_k \}$ and $g_j$ is drawn randomly (according to the Haar measure) from $G$.  

Let $P_{\alpha}(x) = x^{\alpha}$ for all $n$-variate monomials $x^{\alpha}$ of degree at most $d$.  Let $\tau > 0$ be a parameter.  We can compute in $\poly(Q, n^d) $ time, estimates $\wt{P_{\alpha}(x)}$ such that with probability $1- \tau$, we have for all $\alpha$,
\[
\left \lvert \wt{P_{\alpha}} - \E_{g \sim G}\left[ w_1P_{\alpha}(g \cdot x_1) + \dots + w_kP_{\alpha}(g \cdot x_k)   \right] \right \rvert \leq c_d \sigma^d \sqrt{\frac{\log n/\tau}{Q}}
\]
where $c_d$ is a constant depending only on $d$.
\end{lemma}

Copying the proof of Lemma \ref{lem:est-invariant}, we get the following result for estimating the invariant polynomials when our observations come from a mixture.
\begin{lemma}\label{lem:est-mixture-invariant}
Let $G$ be a compact group acting linearly on a vector space $V = \C^n$.  Let $x_1, \dots , x_k \in V$ and assume $\norm{g \cdot x}_2 \leq 1$ for all $g \in G$.  Let $w_1, \dots , w_k$ be mixing weights summing to $1$.  Assume we are given $Q$ independent observations $y_1, \dots , y_Q$ of the form
\[
y_j = g_j \cdot x_l + N(0, \sigma^2 I) + iN(0, \sigma^2 I)
\]  
where $l$ is sampled from $[k]$ according to $\{w_1, \dots , w_k \}$ and $g_j$ is drawn randomly (according to the Haar measure) from $G$.  

Let $\eps$ be a desired accuracy parameter and $\tau$ be the allowable failure probability.  If
\[
Q \geq O_d(1) \poly\left( n^d, \sigma^d, \frac{1}{\eps}, \log \frac{1}{\tau}\right)
\]
the for any invariant polynomial $P$ of degree at most $d$ with coefficients of magnitude at most $1$, we can compute in $\poly(Q)$ time, an estimate $\wt{P}$ such that with probability $1 - \tau$,
\[
\left \lvert \wt{P} - (w_1P(x_1) + \dots + w_kP(x_k)) \right \rvert \leq \eps  \,.
\]
\end{lemma}

\subsection{Decoupling Moments of a Mixture}\label{sec:decouple}
The key observation is that if $P(x)$ is an invariant polynomial, then $P(x)^t$ is also an invariant polynomial for any $t \in \N$.  By measuring 
\begin{align*}
w_1P(x_1) + &\dots + w_k P(x_k) \\
&\vdots \\
w_1P(x_1)^t + &\dots + w_k P(x_k)^t
\end{align*}
for sufficiently large $t$, we may then hope to solve for the individual values of $P(x_1), \dots , P(x_k)$.  This motivates the following lemma.

\begin{lemma}\label{lem:decouple-moments}
Let $ 0 < \eps < 0.5$ be a parameter. Let $z_1, \dots , z_k \in \C$ with $|z_j| \leq K$, $|z{j_1} - z_{j_2} | \geq \eta$ for some constants $K \geq 1, 0 < \eta < 1$.  Let $w_1, \dots , w_k \geq w_{\min}$ be nonnegative real numbers with $w_1 + \dots + w_k = 1$.  Then given estimates $M_j$ for $j = 1,2, \dots , 2k - 1$ with 
\[
\left \lvert M_j - (w_1z_1^j + \dots + w_kz_k^j) \right \rvert \leq \eps \left( w_{\min} (\eta/K)^{k}\right)^{O(1)}
\]
we can compute estimates $\wt{w_1}, \dots  , \wt{w_k}, \wt{z_1}, \dots , \wt{z_k}$ such that there is a permutation $\pi$ on $[k]$ with
\[
|\wt{w_j} - w_{\pi(j)} | + |\wt{z_j} - z_{\pi(j)} | \leq \eps
\]
for all $j \in [k]$.
\end{lemma}

Before we prove Lemma \ref{lem:decouple-moments}, we need the following result about the condition number of a Vandermonde matrix.
\begin{claim}\label{claim:vandermonde}
Let $z_1, \dots , z_k \in \C$ with $|z_j| \leq K$, $|z_{j_1} - z_{j_2} | \geq \eta$ for some constants $K \geq 1, 0 < \eta < 1$.  Let $A$ be the matrix whose rows are $(1, z_j, \dots , z_j^{k-1}) $ for $j = 1,2, \dots , k$.  Then the smallest singular value of $A$ is at least $\frac{1}{k} \cdot \left(\frac{\eta}{2K}\right)^{k-1}$.
\end{claim}
\begin{proof}
Let $s^{(1)}_j$ be the $j$\ts{th} elementary symmetric polynomial in the variables $z_2, \dots , z_k$ i.e.
\begin{align*}
s^{(1)}_{k-1} &= z_2 \cdots z_k \\
&\vdots  \\
s^{(1)}_{1} & = z_2 + \dots + z_k \,.
\end{align*}
We will also use the convention $s^{(1)}_{0} = 1$.  Now consider the vector 
\[
s^{(1)} = ((-1)^{k-1}s^{(1)}_{k-1}, \dots , -s^{(1)}_{1}, s^{(1)}_{0}) \,.
\]
Note that $s^{(1)}A = ( (z_1 - z_2) \cdots (z_1 - z_k)  , 0 \dots , 0)$.  Similarly, we can construct vectors $s^{(2)}, \dots , s^{(k)}$ and let $S$ be the matrix with these vectors as rows.  Then $SA$ is a diagonal matrix with entries 
\[
((z_1 - z_2) \cdots (z_1 - z_k) , \dots ,   (z_k - z_1) \cdots (z_k - z_{k-1}) ) \,.
\]
In particular, all singular values of $SA$ are at least $\eta^{k-1}$.  On the other hand
\[
\norm{S}_{\textsf{op}} \leq \norm{S}_F \leq k \cdot (2K)^{k-1} \,.
\]
Thus, we deduce that the smallest singular value of $A$ is at least $\frac{1}{k} \cdot \left(\frac{\eta}{2K}\right)^{k-1}$.
\end{proof}

Now we prove Lemma \ref{lem:decouple-moments}.
\begin{proof}[Proof of Lemma \ref{lem:decouple-moments}]
Let $A$ be the matrix whose rows are $(1, z_j, \dots , z_j^{k-1})$.  Construct the following $k \times k$ matrices: $M^{(0)}$ has entries $M^{(0)}_{ij} = M_{i + j - 2}$ and $M^{(1)}$ has entries $M^{(1)}_{ij} = M_{i + j - 1}$.  Note that if our estimates were exactly correct, then we would have
\begin{align*}
M^{(0)}_{\textsf{truth}} &= A^T \textsf{Diag}(w_1, \dots , w_k ) A\\
M^{(1)}_{\textsf{truth}} &= A^T \textsf{Diag}(w_1z_1 , \dots , w_kz_k ) A \,.
\end{align*}
Let 
\[
M_{\textsf{truth}} = M^{(1)}_{\textsf{truth}} \left(M^{(0)}_{\textsf{truth}}\right)^{-1} = A^T \textsf{Diag}(z_1 , \dots , z_k ) \left( A^T \right)^{-1} \,.
\]
Note that the eigenvalues of $M_{\textsf{truth}}$ are precisely $z_1, \dots , z_k$.  Now by Claim \ref{claim:vandermonde} and the assumption about our estimates, we can compute an estimate $M = M^{(1)}\left( M^{(0)} \right)^{-1}$ such that  
\[
\norm{M - M_{\textsf{truth}}}_F \leq \eps \left(w_{\min}(\eta/K)^{k}\right)^{O(1)} \,.
\]
Now, we can compute the eigenvalues of $M$.  By Gershgorin's disk theorem (see \cite{moitra2018algorithmic}) we will obtain estimates $\wt{z_1}, \dots , \wt{z_k}$ such that there is a permutation $\pi$ with 
\[
|\wt{z_j} - z_{\pi(j)} | \leq \eps \left(w_{\min}(\eta/K)^{k}\right)^{O(1)}
\]
for all $j$.  Now we can solve for the weights by simply solving a linear system.  Let $\wt{A}$ be the matrix whose rows are  $(1, \wt{z_j}, \dots , \wt{z_j}^{k-1})$.  We solve
\[
\arg\min_w \norm{ w\wt{A} - (M_0, M_1, \dots , M_{k-1}) }_2^2 \,.
\]
Since Claim \ref{claim:vandermonde} gives a bound on the condition number of $A$, we immediately get a similar bound on the condition number of $\wt{A}$.  If we replaced $\wt{A}$ with $A$ and our estimates $M_0 ,\dots , M_{k-1}$ were exactly correct, then the quantity would be minimized when $w = (w_{\pi(1)}, \dots , w_{\pi(k)} )$.

Thus, the solution that we obtain, say $(\wt{w_1}, \dots, \wt{w_k} )$ must satisfy 
\[
|\wt{w_j} - w_{\pi(j)} | \leq \eps \left(w_{\min}(\eta/K)^{k}\right)^{O(1)}
\]
and we are done.
\end{proof}

\subsection{Proof of Theorem \ref{thm:heterogeneous}}

We are now ready to prove Theorem \ref{thm:heterogeneous}.  Lemma \ref{lem:decouple-moments} combined with Lemma \ref{lem:est-mixture-invariant} allows us to recover the values of $ \{ P(f^{[1]}) , \dots , P(f^{[k]}) \}$ for any invariant polynomial $P$.  The main piece that remains is to show how to align two sets of values   $\{ P(f^{[1]}) , \dots , P(f^{[k]}) \},  \{ Q(f^{[1]}) , \dots , Q(f^{[k]}) \}$ for two different invariant polynomials $P, Q$.  To do this, note that $PQ$ is also an invariant polynomial so we can also obtain a set of values $\{ PQ(f^{[1]}) , \dots , PQ(f^{[k]}) \}$ and we will show that since the coefficients of $f^{[1]}, \dots , f^{[k]}$ are smoothed, with high probability there will be a unique way to align the sets $\{ P(f^{[1]}) , \dots , P(f^{[k]}) \}$ and  $\{ Q(f^{[1]}) , \dots , Q(f^{[k]}) \}$.

\begin{proof}[Proof of Theorem \ref{thm:heterogeneous}]
Let $P_1, \dots , P_n$ be the invariant polynomials that we need to measure (see the discussion proceeding the statement of Theorem \ref{thm:heterogeneous}).  Note $n = \poly(NT)$.  Recall that they all have degree at most $C$, coefficients with magnitude at most $C$, and that they all have some leading coefficient of magnitude at least $1/(C poly(N))$ where $C$ is an absolute constant. 

We first prove the following property.  With probability $0.99$ over the random smoothing, we have that for all indices $i, j_1, j_2 \in [k]$ with $j_1 \neq j_2$ and all $a,b \leq n$, that 
\begin{equation}\label{eq:distinct}
\left \lvert P_a(f^{[i]}) P_b(f^{[i]}) - P_a(f^{[j_1]})P_b(f^{[j_2]}) \right \rvert  \geq (\delta/(CNTk))^{O(1)} \,.
\end{equation}
To see this first fix $i, j_1, j_2,a,b$.  Without loss of generality $j_1 \neq i$.  Then we can sample the smoothing of $f^{[j_2]}$ and $f^{[i]}$ first.  By Corollary \ref{coro:anti-concentration-general}, with probability $1 - (\delta/(CNTk))^{10}$, we have $|P_b(f^{[j_2]}) | \geq (\delta/(CNTk))^{O(1)}$.  We can then view  the above as a polynomial in the spherical harmonic coefficients of  $f^{[j_1]}$ and apply Corollary \ref{coro:anti-concentration-general} again.  Union bounding over all $i, j_1, j_2,a,b$, we get the desired conclusion.

Now, we can apply Lemma \ref{lem:est-mixture-invariant} to estimate the quantities 
\begin{align*}
&w_1 P_a(f^{[1]})^t + \dots + w_k P_a(f^{[k]})^t \quad \forall a \in [n] \\
&w_1 \left(P_a(f^{[1]})P_b(f^{[1]}) \right)^t + \dots + w_k  \left(P_a(f^{[k]})P_b(f^{[k]}) \right)^t \quad \forall a,b \in [n] 
\end{align*}
for all $t = 1,2, \dots , 2k$ since products of invariant polynomials are still invariant polynomials.  We can then apply Lemma \ref{lem:decouple-moments} to obtain estimates for the sets 
\begin{align*}
&\{    w_1,  P_a(f^{[1]}), \dots , w_k, P_a(f^{[k]}) \} \quad \forall a \in [n] \\
&\{ w_1, P_a(f^{[1]})P_b(f^{[1]}), \dots ,  w_k, P_a(f^{[k]})P_b(f^{[k]})\}  \quad \forall a,b \in [n]
\end{align*}
that are accurate to say $(\gamma \delta/(CNTk))^{K}$ for some sufficiently large absolute constant $K$ where 
\[
\gamma  = \left( \delta/N\right)^{O(\log N)} (\eps/T)^{O(1)} \,.
\]
Then by (\ref{eq:distinct}), there is a unique way to align them, so we can recover 
\[
\{ \{w_1, P_1(f^{[1]}), \dots ,P_n(f^{[1]}) \}, \dots ,  \{w_k, P_1(f^{[k]}), \dots ,P_n(f^{[k]}) \} \}
\]
to accuracy $\gamma$ up to permutation on $[k]$.  We can then run the algorithm in Theorem \ref{thm:multiple-shells} to recover the functions and we are done. 
\end{proof}

\section{Further Discussion}\label{sec:further-discussion}

\subsection{Equivalence to Tensor Decomposition with Group Structure}

The problem of orbit recovery over $SO(3)$ is closely related to the problem of tensor decomposition over a continuous group (see \cite{moitra2019spectral} for a more complete exposition on these types of problems).  We will go back to the simplest setting (recall Section \ref{sec:problem-formulation}) where there is only one shell and the samples are homogeneous.  For a generic group $G$ acting linearly on a vector space $V = \C^n$, recall that the tensor decomposition problem over the group $G$ can be defined as follows: we are given some $n \times n \times n$ tensor 
\begin{equation}\label{eq:orbit-tensor}
 T = \int_{g \sim G} (g \cdot x)^{\otimes 3}
\end{equation}
where $x$ is some unknown vector and the integral is with respect to the Haar measure of $G$.  The goal is to recover $x$ from $T$.

For finite groups, it is often possible to just treat $T$ as a rank-$|G|$ tensor and employ standard tensor decomposition techniques such as Jennrich's algorithm without using the group structure at all.  Of course, this approach fails for infinite groups and we must instead exploit the group structure.  The main result of this paper, Theorem \ref{thm:main-oneshell}, immediately implies an algorithm for tensor decomposition over $SO(3)$.  The formulation of tensor decomposition over $SO(3)$ is as follows.  There is some vector $ x \in \C^{(N+1)^2}$ and  we observe
\begin{equation}\label{eq:so3-tensor}
T = \int_{R \sim SO(3)} (R \cdot x)^{\otimes 3}
\end{equation}
and our goal is to recover $x$.  To see the similarity to (\ref{eq:sample-model}), we view $x$ as the spherical harmonic coefficients of degree at most $N$ of some function $f$ and the rotation $R$ acts by rotating $f$ and then computing the resulting coefficients.  We will assume that the entries of $x$ are $\delta$-smoothed (i.e. we add Gaussian noise with variance $\delta^2$ to both the real and imaginary part).  As a consequence of Theorem \ref{thm:main-oneshell}, we have the following result:
\begin{theorem}\label{coro:orbit-tensor}
Let $x \in \C^{(N+1)^2}$ be some vector whose entries are $\delta$-smoothed and assume $\norm{x}_2 \leq 1$.  Assume we are given access to a tensor $\wh{T}$ such that
\[
\norm{\wh{T} - \int_{R \sim SO(3)} (R \cdot x)^{\otimes 3}}_F \leq \left(\left(\frac{\delta}{N}\right)^{\log N} \cdot \eps\right)^{O(1)} 
\]
then there is an algorithm that runs in $\poly((N/\delta )^{\log N}/\eps)$ time and with probability $0.9$ (over the smoothing) outputs a vector $\wt{x}$ such that
\[
\norm{ \int_{R \sim SO(3)} (R \cdot \wt{x})^{\otimes 3} -\int_{R \sim SO(3)} (R \cdot x)^{\otimes 3}}_F  \leq \eps \,.
\]
\end{theorem}
\begin{proof}
We can essentially imitate the proof of Theorem \ref{thm:main-oneshell} .  Note that Algorithm \ref{alg:iterate} for iteratively recovering the coefficients only uses the samples to measure the values of the degree-$3$ invariant polynomials.  In this setting, we can simply use $\wh{T}$ to estimate these polynomials to the same accuracy. Instead of Algorithm \ref{alg:const}, we can simply grid search for all possible constant-degree coefficients and then run Algorithm \ref{alg:iterate} to extend each guess (instead of trying to narrow down to a unique guess).  Then at the end, it suffices to compute 
\[
\int_{R \sim SO(3)} (R \cdot \wt{x})^{\otimes 3}
\]
and check if it is indeed close to $\wh{T}$.
\end{proof}

\subsection{What Happens for Cryo-EM?}\label{sec:cryo-EM2}

Cryo-electron microscopy (cryo-EM) is a well-known extension of the problem studied here, cryo-ET \cite{singer2018mathematics}.  In cryo-EM, there is still an unknown function $f: S^2 \rightarrow \C$ but instead of observing a function $\wh{f} = R(f) + \zeta$, we only observe a projection of $\wh{f}$ onto some plane.  While we will not go into the details, it is shown in \cite{bandeira2018estimation} that the degree-$3$ polynomials whose values we can measure (these are analogs of the invariant polynomials) have the following form.
\begin{theorem}[\cite{bandeira2018estimation}]
The degree-$3$ polynomials that we can measure in cryo-EM are 
\[
\mcl{P}_{k_1,k_2, k_3}(f) = \sum_{\substack{l_1,l_2,l_3 \\ |l_1 - l_2| \leq l_3 \leq |l_1 + l_2|}}C_{l_1,k_1,l_2,k_2,l_3,k_3} \langle l_1k_1l_2k_2|l_3(-k_3) \rangle \mcl{I}_{l_1,l_2,l_3}(f)
\]
where $\mcl{I}_{l_1,l_2,l_3}$ are as defined in Theorem \ref{thm:explicit-invariants} and $C_{l_1,k_1,l_2,k_2,l_3,k_3}$ are constants.
\end{theorem}

Note that there are only $O(N^2)$ such polynomials (since we must have $k_1 + k_2 + k_3 = 0$) if we assume that $f$ has spherical harmonic expansion of degree at most $N$, compared to $O(N^3)$ in cryo-ET (recall Theorem \ref{thm:explicit-invariants}).  There are roughly $ N^2$ variables that we need to solve for so the number of constraints and number of variables are comparable.  Also, the polynomials do not have a layered structure.  In particular, for any $k_1,k_2,k_3$, even $k_1 = k_2 = k_3 = 0$, the degree $N$ spherical harmonic coefficients $f_{Nm}$ are involved in the expression for $\mcl{P}_{k_1,k_2, k_3}(f)$.  Thus, it does seem that generalizing Theorem \ref{thm:main-oneshell} to cryo-EM likely requires different techniques.

\bibliographystyle{alpha}
\bibliography{bibliography}

\appendix
\newpage
\begin{center}
\Large{\textbf{Appendix}}
\end{center}

\section{Omitted Proof of Lemma \ref{lem:alg1-analysis}} \label{appendix:const-degree}
In this section, we prove Lemma \ref{lem:alg1-analysis}. The proof will make use of the following tool from algebraic geometry \cite{solerno1991effective}.

\begin{theorem}[Theorem 7 in \cite{solerno1991effective}]\label{thm:dist-to-solution}
Let $f_1, \dots , f_s \in \R[x_1, \dots , x_n]$ and let $D = \sum_{j = 1}^s \deg(f_i)$.  Let $V = \{x \in \R^n: f_1(x) = 0, \dots , f_s(x) = 0 \}$ and assume that $V$ is nonempty.  Then there is a constant $c$ not depending on the $f_i$ and a positive integer $m$ and constant $c'$ (both depending on the $f_i$) such that 
\[
d(x,V)^m \leq c' \cdot (1 + |x|)^{D^{n^{c}}}\max_{j} |f_j(x)|
\]
for all $x \in \R^n$ where $d(x,V)$ denotes the minimum distance from $x$ to an element of $V$.
\end{theorem}

\begin{proof}[Proof of Lemma \ref{lem:alg1-analysis}]
Recall that $C$ is a universal constant.  By Fact \ref{fact:invariant-poly} (property 2), the ring 
\[
\C^{SO(3)}[x_{00}, \dots , x_{C(-C)}, \dots , x_{CC}]
\]
is finitely-generated.  Thus, we can compute the generators $P_1, \dots , P_k$ in $O(1)$ time since these are just fixed polynomials independent of the problem parameters.  This can be done using standard techniques, see e.g. Section 8.1 in \cite{bandeira2018estimation}.

Let $K$ be a sufficiently large universal constant (chosen in terms of $C, P_1, \dots , P_k$).  By Lemma \ref{lem:est-invariant}, with probability $1 - 2^{-10/\gamma}$, our estimates $\wt{P_j}$ for $P_j(\{f_{lm} \}_{l \leq C}) $ all satisfy
\begin{equation}\label{eq:estbound}
\left \lvert \wt{P_j} - P_j(\{f_{lm} \}_{l \leq C}) \right \rvert \leq 0.1 (1/\gamma)^K \,. 
\end{equation}
Note that $\poly(\sigma, \gamma)$ samples suffices because the polynomials $P_1, \dots , P_k$ are fixed, independent of the parameters of the problem.  Now let us write each spherical harmonic coefficient $f_{lm} = a_{lm} + ib_{lm}$ as a sum of its real and imaginary part (we need to do this because Theorem \ref{thm:dist-to-solution} is for polynomials with real coefficients).  We can also decompose each polynomial $P_j$ into its real and imaginary parts, i.e.
\[
P_j\left(\{ a_{lm} + ib_{lm}\}_{l \leq C} \right) = \textsf{Re}_j(\{a_{lm}, b_{lm} \}_{l \leq C} )  + i \textsf{Im}_j(\{a_{lm}, b_{lm} \}_{l \leq C} )
\]
where $\textsf{Re}_j(\{a_{lm}, b_{lm} \}_{l \leq C} ), \textsf{Im}_j(\{a_{lm}, b_{lm} \}_{l \leq C} )$ are each polynomials in the variables $\{a_{lm}, b_{lm} \}_{l \leq C}$ with real coefficients.  Now consider the system $\mcl{S}$ defined as follows:
\begin{align*}
\textsf{Re}_j(\{\wt{a_{lm}}, \wt{b_{lm}} \}_{l \leq C} ) = \textsf{Re}_j(\{\wh{a_{lm}}, \wh{b_{lm}} \}_{l \leq C} ) \quad \forall j \in [k]\\
\textsf{Im}_j(\{ \wt{a_{lm} }, \wt{b_{lm}} \}_{l \leq C} ) = \textsf{Im}_j(\{ \wh{a_{lm} }, \wh{b_{lm}} \}_{l \leq C} ) \quad \forall j \in [k]
\end{align*}
where the variables are $\{\wt{a_{lm}}, \wt{b_{lm}} \}_{l \leq C},  \{ \wh{a_{lm} }, \wh{b_{lm}} \}_{l \leq C} $ i.e. there are $2(C+1)^2$ variables.  By Fact \ref{fact:invariant-poly}  (part 3), the solutions to this system are precisely the sets of $\{\wt{a_{lm}}, \wt{b_{lm}} \}_{l \leq C} ,  \{\wh{a_{lm}}, \wh{b_{lm}} \}_{l \leq C}$ with the following property: there is a rotation $R \in SO(3)$ such that 
\[
\{ \wt{a_{lm}} + i \wt{b_{lm}} \}_{l \leq C} = R(\{\wh{a_{lm}} + i \wh{b_{lm}} \}_{l \leq C})
\]
i.e. the coefficients $\{ \wt{a_{lm}} + i \wt{b_{lm}} \}_{l \leq C}$ and $\{\wh{a_{lm}} + i \wh{b_{lm}} \}_{l \leq C}$ are equivalent up to rotation.

If our guesses for $\wt{f_{lm}}$ satisfy 
\begin{equation}\label{eq:test}
\left \lvert P_j( \{\wt{f_{lm}}\}_{l \leq C}) - \wt{P_j}\right \rvert  \leq 0.2 (1/\gamma)^{K} \quad \forall j \in [k]
\end{equation}
then by (\ref{eq:estbound}) we have
\[
\left \lvert P_j( \{\wt{f_{lm}}\}_{l \leq C}) - P_j( \{ f_{lm}\}_{l \leq C}) \right \rvert \leq  (1/\gamma)^{K} \quad \forall j \in [k] \,.
\]
Then by Theorem \ref{thm:dist-to-solution} applied to the system $\mcl{S}$,  there must be a rotation $R \in SO(3)$ and coefficients $\{\wt{f_{lm}}'\}_{l \leq C}, \{f_{lm}'\}_{l \leq C}$ such that 
\begin{align*}
&\{\wt{f_{lm}}'\}_{l \leq C} = R(\{f_{lm}'\}_{l \leq C}) \\ 
&\sum_{l = 0}^C \sum_{m = -l}^l (f_{lm} - f_{lm}')^2 + \left(\wt{f_{lm}} - \wt{f_{lm}}'\right)^2 \leq \gamma 
\end{align*}
where we use that $K$ is a sufficiently large universal constant.  In other words, $\{\wt{f_{lm}}\}_{l \leq C}$ and $\{ f_{lm}\}_{l \leq C} $ are close to some pair of sets of coefficients that are equivalent up to rotation.  Next by Fact \ref{fact:orthonormality} (orthonormality of spherical harmonics), 
\begin{align*}
&\gamma \leq \sum_{l = 0}^C \sum_{m = -l}^l\left(f_{lm} - f_{lm}'\right)^2  = \norm{f_{\leq C} - f'_{\leq C}}_2^2 = \norm{R(f)_{\leq C} - R(f')_{\leq C}}_2^2 = \norm{ R(f)_{\leq C} - \wt{f}'_{\leq C}}_2^2 \\
&\gamma \leq \sum_{l = 0}^C \sum_{m = -l}^l\left(\wt{f_{lm}} - \wt{f_{lm}}'\right)^2 = \norm{\wt{f}_{\leq C} - \wt{f}'_{\leq C} }_2^2
\end{align*}
so we deduce
\[
\norm{\wt{f}_{\leq C} - R(f)_{\leq C}}_2^2 = \sum_{l = 0}^C \sum_{m = -l}^l |\wt{f_{lm}} - R(f)_{lm}|^2 \leq 4\gamma \,.
\]
It remains to show that with high probability, one of our guesses actually satisfies the test (\ref{eq:test}).  This is true simply because the polynomials $P_1, \dots , P_k$ are fixed and $|f_{lm}| \leq 1$ for all $l,m$ so as long as we grid search with a sufficiently fine grid i.e. $(1/\gamma)^{K_0 \cdot K}$ for sufficiently large constant $K_0$, then the guess $\{\wt{f_{lm}} \}_{l \leq C}$ that is entrywise closest to $\{f_{lm} \}_{l \leq C}$ must satisfy 
\[
\left \lvert P_j( \{\wt{f_{lm}}\}_{l \leq C}) -  P_j( \{ f_{lm}\}_{l \leq C}) \right \rvert  \leq 0.1 (1/\gamma)^{K} \quad \forall j \in [k]
\]
which combined with (\ref{eq:estbound}) means that this guess passes the test.  Overall, it is clear that the algorithm runs in time $\poly(\sigma, \gamma)$ and the proof is complete.
\end{proof}

\section{Quantitative Bounds on Polynomials}\label{sec:poly-anticoncentration}
In smoothed analysis, it is usually necessary to show that a bad event e.g. some matrix being very close to singular, occurs with low probability.  This is usually done by proving anticoncentration of various quantities.  We begin with a standard anticoncentration bound for polynomials (see e.g. \cite{carbery2001distributional}).  Its proof is included here for the sake of completeness. 
\begin{claim}\label{claim:univariate-poly}
Let $P(x) : \R \rightarrow \R$ be a polynomial in one variable of degree at most $d$ with leading coefficient $1$.  Then
\[
\mu \{|P(x)| < \delta \} \leq 20\delta^{1/d}
\]
for all positive real numbers $\delta$ where $\mu$ denotes the uniform measure on the real line.
\end{claim}
\begin{proof}
Let $S = \{|P(x)| < \delta \} $.  We claim that there cannot be real numbers $x_1 <  \dots < x_{d+1}$ in $S$ such that 
\[
|x_i - x_j| \geq  \frac{10}{d} \delta^{1/d}
\]
for all $i \neq j$.  To see this, by the Lagrange Interpolation formula
\[
P(x) = \sum_{i} \frac{P(x_i) \prod_{i' \neq i} (x - x_{i'}) }{\prod_{i' \neq i}(x_i - x_{i'})} \,. 
\]
The leading coefficient on the RHS of the above is less than 
\[
\frac{d^d}{10^d}\sum_{i=1}^{d+1}  \frac{1}{\prod_{i' \neq i} |i' - i| } = \frac{d^d}{10^d}\frac{2^{d}}{d!} < 1
\]
which is a contradiction.  Thus, if we take a maximal set of $ (10/d) \delta^{1/d}$-separated points in $S$, we conclude that the measure of $S$ is at most 
\[
2d \frac{10}{d} \delta^{1/d} = 20\delta^{1/d}
\]
\end{proof}

The lemma that is most important in our proof gives an anticoncentration inequality for multivariate polynomials that is somewhat nonstandard.  While multivariate generalizations of Claim \ref{claim:univariate-poly} exist (see e.g. \cite{carbery2001distributional, meka2015anticoncentration}), the key difference is that those results only use information about the total degree of the polynomial $P(x_1, \dots , x_n)$.  When only using the total degree, of course it is not possible to beat the bound in Claim \ref{claim:univariate-poly}.  However, in the next result, we prove that when the degrees in the $n$ variables are somewhat balanced, we can obtain a much better anticoncentration bound that has failure probability exponentially small in $n$.

\begin{lemma}\label{lemma:multivariate-poly}
Let $P(x_1, \dots , x_n): \R^n \rightarrow \R$ be a polynomial in $n$ variables.  Assume that for some nonnegative integers $a_1, \dots , a_n$
\begin{itemize}
    \item The coefficient of the monomial $x_1^{a_1}x_2^{a_2} \dots x_n^{a_n}$ in $P$ is equal to $1$
    \item For each $i$, the degree of $P$ when viewed as a polynomial in only $x_i$ is equal to $a_i$
\end{itemize}
then for any given real numbers $\alpha_1, \dots , \alpha_n$, if we perturb them to $\alpha_1 + \delta_1 , \dots , \alpha_n + \delta_n $ where the $\delta_i$ are drawn independently at random from $N(0,\delta^2)$ then for any $\eps > 0$,
\[
|P(\alpha_1 + \delta_1, \dots , \alpha_n + \delta_n) | \geq \delta^{a_1 + \dots + a_n}\left(\frac{\eps}{e}\right)^{40(a_1 + \dots + a_n)}
\]
with probability at least 
\[
1 - \eps^{20(a_1 + \dots + a_n)/\max(a_1, \dots, a_n)} \,.
\]
\end{lemma}
\begin{proof}
Note that by rescaling the polynomial $P$ and dividing out the leading coefficient, we may assume that $\delta = 1$.  Let $d = a_1 + \dots + a_n$ and let $f = 2\max(a_1, \dots , a_n)$.  

For each $i$ with $1 \leq i \leq n$, let  $P_{-i}[y_1, \dots , y_i]$ be the polynomial $P$, viewed as a polynomial in variables $x_{i+1}, \dots , x_n$ after plugging in values $y_1, \dots , y_i$ for  $x_1, \dots , x_i$ respectively.

Let $L_{-i}[y_1, \dots , y_i]$ be the coefficient of $x_{i+1}^{a_{i+1}} \cdots x_n^{a_n}$ in $P_{-i}[y_1, \dots , y_i]$.  We will plug in values for the variables $x_1, \dots , x_n$ in that order and analyze how the sequence $L_{-1}, \dots , L_{-n}$ behaves.  We first show that once $\delta_1, \dots , \delta_i$ are fixed, we have
\begin{equation}\label{eq:martingale-bound}
\E\left[\frac{1}{(L_{-(i+1)}[\alpha_1 + \delta_1 , \dots , \alpha_{i+1} + \delta_{i+1}])^{1/f}} \right] \leq \frac{ e^{40a_{i+1}/f}}{(L_{-i}[\alpha_1 + \delta_1 , \dots , \alpha_{i} + \delta_{i}])^{1/f}}
\end{equation}
where the expectation on the LHS is over the randomness in the choice of $\delta_{i+1}$.  To see this, note that once $\delta_1, \dots , \delta_i$ are chosen, $ L_{-(i+1)}[\alpha_1 + \delta_1 , \dots , \alpha_{i+1} + \delta_{i+1}]$ is obtained by plugging in $\alpha_{i+1} + \delta_{i+1}$ into some degree $a_{i+1}$ polynomial in $x_{i+1}$ with leading coefficient $L_{-i}[\alpha_1 + \delta_1 , \dots , \alpha_{i} + \delta_{i}]$.  By Claim \ref{claim:univariate-poly},
\begin{align*}
& E_{\delta_{i+1}}\left[\frac{1}{(L_{-(i+1)}[\alpha_1 + \delta_1 , \dots , \alpha_{i+1} + \delta_{i+1}])^{1/f}} \right] \\ &\leq \frac{1 + \int_{1}^{\infty} \Pr\left[\frac{1}{(L_{-(i+1)}[\alpha_1 + \delta_1 , \dots , \alpha_{i+1} + \delta_{i+1}])^{1/f}} \leq \frac{x}{ (L_{-i}[\alpha_1 + \delta_1 , \dots , \alpha_{i} + \delta_{i}])^{1/f} } \right] dx}{(L_{-i}[\alpha_1 + \delta_1 , \dots , \alpha_{i} + \delta_{i}])^{1/f}} 
\\ & \leq \frac{1}{(L_{-i}[\alpha_1 + \delta_1 , \dots , \alpha_{i} + \delta_{i}])^{1/f}} \left(1 + \int_1^{\infty}  20 \left(\frac{1}{x}\right)^{f/a_{i+1}} dx \right) \,.
\end{align*}
 Note 
\[
1 + \int_1^{\infty}  20 \left(\frac{1}{x}\right)^{f/a_{i+1}} dx  = 1 + \frac{20a_{i+1}/f}{1 - a_{i+1}/f} \leq 1 + 40a_{i+1}/f \leq e^{40a_{i+1}/f} \,,
\]
and this establishes equation (\ref{eq:martingale-bound}).  

Now we can multiply equation (\ref{eq:martingale-bound}) over all $i$ to get that 
\[
\E\left[\frac{1}{(L_{-n}[\alpha_1 + \delta_1 , \dots , \alpha_{n} + \delta_{n}])^{1/f}} \right] \leq e^{40 d / f} \,.
\]
Note that $L_{-n}[\alpha_1 + \delta_1 , \dots , \alpha_{n} + \delta_{n}]$ is exactly the value of $P(\alpha_1 + \delta_1 , \dots , \alpha_{n} + \delta_{n})$.  By Markov's inequality, we deduce that this value is less than $(\eps/e)^{40d}$ with probability at most $\eps^{40d/f}$, completing the proof.

\end{proof}

Since our main proof works over $\C$, we will translate the above result to work over $\C$.  The proof follows simply by separating real and imaginary parts and applying the previous lemma.
\begin{corollary}\label{coro:multivariate-anticoncentration}
Let $P(x_1, \dots , x_n): \C^n \rightarrow \C$ be a polynomial in $n$ variables.  Assume that for some nonnegative integers $a_1, \dots , a_n$
\begin{itemize}
    \item The coefficient of the monomial $x_1^{a_1}x_2^{a_2} \dots x_n^{a_n}$ in $P$ has magnitude $1$
    \item For each $i$, the degree of $P$ when viewed as a polynomial in only $x_i$ is equal to $a_i$
\end{itemize}
then for any given complex numbers $\alpha_1, \dots , \alpha_n$, if we perturb them to $\alpha_1 + \delta_1 + \gamma_1 i, \dots , \alpha_n + \delta_n + \gamma_n i $ where the $\delta_j, \gamma_j$ are drawn independently at random from $N(0,\delta^2)$ then for any $\eps > 0$,
\[
|P(\alpha_1 + \delta_1 + \gamma_1 i, \dots , \alpha_n + \delta_n + \gamma_n i) | \geq \delta^{a_1 + \dots + a_n}\left(\frac{\eps}{e}\right)^{40(a_1 + \dots + a_n)}
\]
with probability at least 
\[
1 - \eps^{20(a_1 + \dots + a_n)/\max(a_1, \dots, a_n)} \,.
\]
\end{corollary}
\begin{proof}
We can multiply $P$ by a suitable constant so that the coefficient of $x_1^{a_1}x_2^{a_2} \dots x_n^{a_n}$ is equal to $1$.  Now we can write each variable $x_j = y_j + iz_j$ and split $P$ into its real and imaginary parts.  Then the real part of $P$, say $P_{\text{real}}(y_1, z_1, \dots , y_n, z_n)$ contains the monomial $y_1^{a_1} \cdots y_n^{a_n}$ with coefficient $1$.  Now we can consider sampling $\gamma_1, \dots , \gamma_n$ first and fixing them so that now $z_1, \dots , z_n$ are fixed and $P_{\text{real}}$ becomes a polynomial in $n$ variables $y_1, \dots , y_n$ with individual degrees at most $a_1, \dots , a_n$ respectively.  We now apply Lemma \ref{lemma:multivariate-poly} on this polynomial to complete the proof.
\end{proof}

We will also use more standard anti-concentration inequalities in a few of the proofs.  We begin with the well-known Carbery-Wright inequality \cite{carbery2001distributional}.
\begin{lemma}[Carbery-Wright]\label{lem:carbery-wright}
There is a universal constant $B$ such that for any Gaussian $G$ over $\R^n$ and polynomial $P(x_1, \dots , x_n): \R^n \rightarrow \R$,
\[
\Pr_{x \sim G}\left[ |P(x) | \leq \eps \sqrt{\Var_{x \sim G}[P(x)]} \right] \leq B \eps^{1/d} \,.
\]
\end{lemma}

\begin{claim}\label{claim:gaussian-variance}
Let $P(x_1, \dots , x_n): \R^n \rightarrow \R$ be a polynomial of degree $d \geq 1$ such that some monomial $x_1^{a_1}x_2^{a_2} \cdots x_n^{a_n}$ of degree $d$ has coefficient $1$.  Then
\[
\Var_{x \sim N(0, I)}[ P(x)] \geq 1 \,.
\]
\end{claim}
\begin{proof}
Let $H_1(x) = x, H_2(x) = x^2 - 1, H_3(x) = x^3 - 3x, \cdots $ be the Hermite polynomials.  The following properties are well-known:
\[
\E_{x \sim N(0,1)}[H_i(x)] = 0 \quad \forall i \geq 1
\]
\begin{equation}\label{eq:hermite}
\E_{x \sim N(0,1)}[H_i(x)H_j(x)] = 1_{i = j} (i!) \,.
\end{equation}
Now we can write $P$ in the following form
\[
P = \sum_{\alpha} c_{\alpha}H_{\alpha_1}(x_1) \cdots H_{\alpha_n}(x_n)
\]
where $\alpha$ runs over all $n$-tuples corresponding to monomials of degree at most $d$.  Note this decomposition is unique and also for $\alpha = (a_1, \dots , a_n)$, the coefficient $c_{\alpha} = 1$.  Then by (\ref{eq:hermite}),
\[
\Var_{x \sim N(0, I)}[P(x)] = \sum_{\alpha \neq 0} c_{\alpha}^2 \alpha_1! \cdots \alpha_n! \geq 1 \,. 
\]
\end{proof}

Combining the two previous claims, we get:
\begin{corollary}\label{coro:anti-concentration-general}
Let $P(x_1, \dots , x_n): \C^n \rightarrow \C$ be a polynomial of degree $d \geq 1$ such that some monomial $x_1^{a_1}x_2^{a_2} \cdots x_n^{a_n}$ of degree $d$ has coefficient with magnitude $1$.  Then for any given complex numbers $\alpha_1, \dots , \alpha_n$, if we perturb them to $\alpha_1 + \delta_1 + \gamma_1 i, \dots , \alpha_n + \delta_n + \gamma_n i $ where the $\delta_j, \gamma_j$ are drawn independently at random from $N(0,\delta^2)$ then for any $\eps > 0$, 
\[
\Pr\left[ |P(\alpha_1 + \delta_1 + \gamma_1 i, \dots , \alpha_n + \delta_n + \gamma_n i) | \leq \delta^d \cdot \eps  \right] \leq B\eps^{1/d}
\]
for some universal constant $B$.
\end{corollary}
\begin{proof}
Multiply $P$ by a suitable constant so that the coefficient of $x_1^{a_1}x_2^{a_2} \cdots x_n^{a_n}$  is equal to $1$.  Sample $\gamma_1 , \dots , \gamma_n$ first and fix their values.  Now let $Q$ be the polynomial such that 
\[
Q(\delta_1, \dots , \delta_n) = P_{\text{real}}(\alpha_1 + \delta_1 + \gamma_1 i, \dots , \alpha_n + \delta_n + \gamma_n i) \,,
\]
where $P_{\text{real}}$ denotes the real part of $P$.  Note that $Q$ has real coefficients and is real valued and furthermore, the coefficient of $x_1^{a_1}x_2^{a_2} \cdots x_n^{a_n}$ is also equal to $1$.  Now we can apply Claim \ref{claim:gaussian-variance} and Lemma \ref{lem:carbery-wright} on the polynomial $\delta^{-d} Q(\delta x_1, \dots , \delta x_n)$ and get the desired conclusion.
\end{proof}

\end{document}